\RequirePackage{fix-cm}
\documentclass[smallextended]{svjour3}       
\smartqed  
\usepackage{graphicx}
\usepackage{makeidx}
\usepackage{amssymb}
\usepackage{amsmath}
\begin{document}

\title{A Modal Characterization of Alternating Approximate Bisimilarity\thanks{This work received financial support of the National Natural Science of China(No. 60973045), Fok Ying-Tung Education Foundation, the NSF of Jiangsu Province (No.BK2012473) and the PAPD of JiangSu Higher Education Institutions (Audit Science and Technology Preliminary Research Project (No.YSXKKT27))
.}
}


\author{Jinjin Zhang         \and
        Zhaohui Zhu 
}


\institute{Jinjin Zhang \at
School of Information Science, Nanjing Audit University, Nanjing, 211815 China\\
Tel. +86-13770658202\\
\email{jinjinzhang@nau.edu.cn}
\and
Zhaohui Zhu: Corresponding author \at
Department
of Computer Science, Nanjing University of Aeronautics and Astronautics, Nanjing,  210016 China\\
\email{bnj4892856@jlonline.com}
}

\date{Received: date / Accepted: date}

\maketitle

\begin{abstract}
Recently, alternating transition systems are adopted to describe control systems with disturbances and their finite abstract systems. In order to capture the equivalence relation between these systems, a notion of alternating approximate bisimilarity is introduced. This paper aims to establish a modal characterization for alternating approximate bisimilarity.
Moreover, based on this result, we provide a link between  specifications satisfied by the samples of control systems with disturbances and their finite abstractions.
\end{abstract}

\section{Introduction}\label{Int}
The notion of bisimilarity is one of the central concepts in process algebra.
Roughly speaking, two states are bisimilar if and only if they can perform same actions to reach bisimilar states.
In general, two bisimilar processes are always considered to be identical.
In recent years, the notion of bisimilarity has been adopted in the area of control theory to capture the equivalence between control systems and their finite abstraction~\cite{tabuada2003finite}\cite{tabuada2003discrete}\cite{tabuada2006linear}.

However, when the states or actions of labeled transition systems  are associated with quantitative data, the notion of bisimilarity seems not be very suitable for describing the equivalence in such situation.
For example, in real time systems, there is often a little difference between time delays.
If we use the usual notion of bisimilarity (for instance, timed bisimilarity~\cite{larsen1994time}\cite{wang1990real}) to capture the equivalence of states in such systems, time delays can match only if they are identical.
Such exact matching may be unrealistic.
On the other hand, in control theory, it has been pointed out that the notion of bisimilarity is so rigorous  that it is often hard to construct finite abstractions which are bisimilar to the given control systems~\cite{girard2007hybrid}\cite{pola2008approximately}.

To overcome these defects, a number of theories are provided to describe approximate behavioral equivalence~\cite{van2005metricbehavioural}\cite{van2008approximating}\cite{van2005behavioural}\cite{desharnais1999labelled}\cite{desharnais2004metrics}\cite{giaca1990algebraic}\cite{girard2007approximation}
\cite{pola2008symbolic}\cite{pola2009symbolic}\cite{ying2002bisimulation}.
In these work, two different approaches have been adopted.

One approach is to introduce notions of approximate bisimilarity.
In such category, Giacalone et al. are probably first to present the notion of approximate bisimilarity, and provide $\varepsilon$-bisimilarity  for probabilistic transition systems~\cite{giaca1990algebraic}.
In the framework of metric labeled transition systems, Ying provides the notion of $\lambda$-bisimilarity~\cite{ying2002bisimulation}.
This notion has been adopted to describe the equivalence between processes in time-CCS and time-CSP~\cite{ying2002bisimulation}, the equivalence and reliability of processes in pi-calculus with noise~\cite{ying2005pi}, and  the equivalence between quantum processes in qCCS~\cite{ying2009algebra}.

In recent years, some notions of approximate bisimilarity are introduced in control theory.
In the framework of LTS with observations and metrics over observations, Girard and Pappas introduce $\delta$-approximate bisimilarity~\cite{girard2007approximation}.
Pola and Tabuada adopt this notion to capture the equivalence between control systems without disturbances and their finite abstractions~\cite{pola2008approximately}.
They also provide the notion of alternating approximate bisimilarity in alternating transition systems to describe approximate equivalence between control systems with disturbances and their finite abstractions~\cite{pola2008symbolic}\cite{pola2009symbolic}.
The notions of approximate bisimilarity play an important role in the analysis and design of control systems (for example, see~\cite{camara2011synthesis}\cite{girard2009hierarchical}\cite{tabuada2008approximate}).
Girard and Pappas give an overview about the related work and point out that the notions of approximate bisimilarity provide a bridge between control theory and computer science~\cite{girard2011approximate}.

Another approach is based on distance functions over processes (or states, systems).
For a variety  of transition systems, distance functions have been introduced via distinct approaches (e.g., modal logic, fixed point, and coalgebra).
For example, for probabilistic transition systems, Desharnais et al.~\cite{desharnais1999labelled}\cite{desharnais2004metrics} and Breugel et al.~\cite{van2008approximating}\cite{van2005behavioural} adopt these methods to define metrics over processes and establish the relationship between these metrics.
Recently, Zhou and Ying define a metric over probabilistic transition systems in terms of so-called ``smallest'' logical formula that distinguishes them~\cite{zhou2012approximating}.
For labeled transition systems accompanied with metric, van Breugel provides pseudometrics over states through  these three methods and shows that these pseudometrics coincide~\cite{van2005metricbehavioural}.
To describe the equivalence between metric transition systems, de Alfaro et al. introduce linear distances and branching distances~\cite{de2009linear}.

The relationship between these two approaches has been explored in the literature. For example, Giacalone et al.
introduce a pseudometric over probabilistic transition systems in terms of $\varepsilon$-bisimilarity~\cite{giaca1990algebraic}.
Van Breugel presents a conjecture which concerns the relationship between his behavioural pseudometric and Ying's $\lambda$-bisimilarity~\cite{van2005metricbehavioural}. Recently, a negative answer to this conjecture is given~\cite{zhang2008behavioural}.
In the framework of LTS with observations, Girard and Pappas  characterize the branching distance in terms of $\delta$-approximate bisimulation with the assumption that the discount factor $\alpha=1$~\cite{girard2007approximation}.
This result has been generalized to general case and the branching distance with arbitrary discount factor is characterized in terms of $(\eta,\alpha)$-bisimilarity~\cite{zhang2008characterize}.

As well known, bisimilarity can be characterized as a fixed point~\cite{park1981concurrency}, via a modal logic~\cite{hennessy1985algebraic} and by way of coalgebra~\cite{aczel1989final}.
A modal characterization of bisimilarity is provided by Hennessy and Milner~\cite{hennessy1985algebraic}.
They demonstrate that, in the framework of LTS, bisimilarity coincides with logical equivalence w.r.t Hennessy-Milner logic (HML, for short), that is, two states in LTS are bisimilar if and only if they satisfy the same formulae of HML.
Inspired by this result, different modal characterizations are established for a lot of varieties of bisimilarity in the above style.
For instance, Alur et al. characterize alternating bisimilarity in terms of alternating-time temporal logic (ATL, for short)~\cite{AlurCONCUr98}.
Ying provides a logical characterization of $\lambda$-bisimilarity with the assumption that the metric is ultra-metric or $\lambda=0$.
However, without such assumption, $\lambda$-bisimilarity can not be characterized in the style of Hennessy-Milner theorem.
Its logical characterization associated with arbitrary metric is established in a new style~\cite{zhang2007modal}.

The logical characterizations of bisimilarity  play important roles in the formal analysis and design of control systems.
They guarantee that control systems share the same logical properties with their finite abstractions which are bisimilar to these control systems.
In such situation, the analysis and design of control systems can be equivalently performed on their finite abstraction, which considerably reduces the complexity of the analysis and design of control systems~\cite{alur2000discrete,tabuada2006linear}.

This paper aims to establish a logical characterization of alternating approximate bisimilarity.
Furthermore, based on this result, for control systems with disturbances mentioned in~\cite{pola2009symbolic}, we illustrate a relationship between linear temporal logical  specifications satisfied by their samples under control and by their finite abstractions under control, respectively.
Roughly speaking, this paper demonstrates that  if the sample of a control system with disturbances and its finite abstraction are alternating approximate bisimilar and the latter satisfies a specification  under control, then the former may satisfy a ``looser" specification under control.
In particular, the transformation from a given specification to a looser one is provided.

The rest of this  paper is organized as follows.
We recall related definitions and results in Section~\ref{sec:pre}.
In Section~\ref{sec:relation}, we provide a variety of ATL and two relations between the formulas of this logical language, which play central roles in this paper.
Section~\ref{sec:characterization} establishes a modal characterization of alternating approximate bisimilarity.
In  Section~\ref{sec:finite abstraction}, we illustrate a relationship between temporal logical specifications satisfied by the samples of control systems with disturbances and by their finite abstractions under control.
Finally, we  conclude the paper in Section~\ref{sec:conclusion}.
\section{Preliminaries}
\label{sec:pre}
This section will recall some notions and results about alternating transition systems, alternating bisimilarity and alternating approximate bisimilarity from \cite{AlurCONCUr98}\cite{pola2009symbolic}.

Before doing so, we introduce some useful notations. The symbol $\mathbb{N},\mathbb{R},\mathbb{R}_+$ and $\mathbb{R}_+^0$ denote the set of positive integers, reals, positive reals and nonnegative reals, respectively.
For any set $A$, $A^{+}$ represents the set of all non-empty finite strings over $A$, and $A^{\omega}$ denotes the set of infinite strings over $A$.
We use $s_{A}$ and $\sigma_{A}$ to denote the elements of $A^{+}$ and $A^{\omega}$, respectively. If $A$ is known from the context, we will omit the subscripts in $s_{A}$ and $\sigma_{A}$.
For any $s\in A^{+}$, $s[i]$ and $s[end]$ mean the $i$-th element and the last element of $s$, respectively.
Given $i\leq j$, $s[i,j]$, $s[i,end]$ and $\sigma[i,\infty]$ represent $s[i]s[i+1]\cdots s[j]$, $s[i]s[i+1]\cdots s[end]$ and $\sigma[i]\sigma[i+1]\cdots$, respectively. As usual, $|s|$ means the length of $s$.

\subsection{Alternating transition systems}\label{subsec:pre_systems}
\begin{definition}\label{def:ATS}
An alternating transition system is a 5-tuple $(S,\mathbb{P},\Omega,\Pi,\hbar)$, where

$\bullet$ $S$ is a set of states;

$\bullet$ $\mathbb{P}$ is a set of observations;

$\bullet$ $\Omega$ is a finite set of agents;

$\bullet$ $\Pi:S\rightarrow\mathbb{P}$ is an observation function;

$\bullet$ $\hbar:S\times\Omega\rightarrow 2^{2^S}$ is a function satisfying that for any state $q\in S$, if each agent $i\in \Omega$ chooses a set $S_i\in\hbar(q,i)$, then the set $\bigcap_{i\in \Omega}S_i$ is a singleton. The function $\hbar$ is often said to be transition function.

If $\bigcup\hbar(q,i)$ is finite for each $q\in S$ and $i\in\Omega$, then we say $(S,\mathbb{P},\Omega,\Pi,\hbar)$ is finite branching.
If both the state set $S$ and the observation set $\mathbb{P}$ are finite, then $(S,\mathbb{P},\Omega,\Pi,\hbar)$ is said to be a finite alternating transition system.
\end{definition}

Intuitively, for each state $q$, an agent $i$ can choose a set $S_i\in \hbar(q,i)$ such that the state reached from $q$ must belong to $S_i$.
According to the above definition, it is clear that the successor state of state $q$ is determined when all agents make choices.

\begin{definition}\label{def:strategy}
Let $T=(S,\mathbb{P},\Omega,\Pi,\hbar)$ be an alternating transition system,  $i\in \Omega$ and $Ag\subseteq \Omega$. A function $f_i:S^+\rightarrow 2^S$ is said to be a strategy of $i$ iff $f_i(s)\in\hbar(s[end],i)$ for any $s\in S^+$.
A function $F_{A\!g}:S^+\rightarrow 2^S$ is said to be a strategy of $Ag$ iff there exist a family of strategies $f_i$ of $i$ $(i\in Ag$) such that $F_{A\!g}(s)=\bigcap_{i\in Ag}f_i(s)$ for any $s\in S^+$.
\end{definition}

In the following, we set $\hbar(q,Ag)\triangleq\{\bigcap_{i\in Ag}S_i:S_i\in\hbar(q,i) \textrm{ for each } i\in Ag\}$ for each state $q$ and agent set $Ag$. The  conclusion below is simple but useful.

\begin{lemma}\label{lem:stra_Ag}
Let $T=(S,\mathbb{P},\Omega,\Pi,\hbar)$ be an alternating transition system and $Ag\subseteq \Omega$. The function $F_{A\!g}:S^+\rightarrow 2^S$ is a strategy of $Ag$ if and only if $F_{A\!g}(s)\in \hbar(s[end],Ag)$ for any $s\in S^+$.
\end{lemma}
\begin{proof}
(From Left to Right) Immediately follows from Definition~\ref{def:strategy} and the definition of $\hbar$.

(From Right to Left) Suppose that the function $F_{A\!g}:S^+\rightarrow 2^S$ satisfies that $F_{A\!g}(s)\in \hbar(s[end],Ag)$ for any $s\in S^+$. To show that this function is a strategy of $Ag$, construct a family of strategies $f_i(i\in Ag)$ as follows:

Let $s\in S^+$. Since $F_{A\!g}(s)\in \hbar(s[end],Ag)$, it follows from the definition of  $\hbar(s[end],Ag)$ that there exist $Q_i\in \hbar(s[end],i)$ for each $i\in Ag$ such that $F_{A\!g}(s)=\bigcap_{i\in Ag}Q_i$.
Then, fix these $Q_i$'s, and for each $i\in Ag$, we set $f_i(s)=Q_i$.

Clearly, by Definition~\ref{def:strategy}, for each $i\in Ag$, the function $f_i:S^+\rightarrow 2^S$ defined above is a strategy of $i$. On the other hand,  it is easy to check that for any $s\in S^+$, $F_{A\!g}(s)=\bigcap_{i\in Ag}f_i(s)$.
Therefore, it follows from Definition~\ref{def:strategy} that $F_{A\!g}$ is a strategy of $Ag$.
\qed\end{proof}

In general, the strategies are provided for some agents to enforce the outcomes of alternating transition systems to satisfy the given properties, such as reachability, safety and so on.
Formally, the outcomes of alternating transition systems under strategies are defined below.

\begin{definition}\label{def:outcome}
Let $T=(S,\mathbb{P},\Omega,\Pi,\hbar)$ be an alternating transition system, $q\in S$, $Ag\subseteq \Omega$ and let $F_{A\!g}:S^+\rightarrow 2^S$ be a strategy of $Ag$.
For each $n\in\mathbb{N}$,
\begin{equation*}
\begin{aligned} Out_{T}^n(q,F_{A\!g})\triangleq  \{s\in &S^+: s[1]=q, |s|=n \textrm{ and} \\ & \forall i<n \exists Q\in \hbar(s[i],\Omega-Ag)(F_{A\!g}(s[1,i])\cap Q=\{s[i+1]\})\}
\end{aligned}
\end{equation*}
and
\begin{equation*}
\begin{aligned} Out_{T}(q,F_{A\!g})\triangleq  \{\sigma\in &S^\omega: \sigma[1]=q \textrm{ and} \\ & \forall i\in\mathbb{N} \exists Q\in \hbar(\sigma[i],\Omega-Ag)(F_{A\!g}(\sigma[1,i])\cap Q=\{\sigma[i+1]\})\}.
\end{aligned}
\end{equation*}
\end{definition}

Intuitively, $F_{A\!g}$ is used to indicate a family  of choices of agent set $Ag$, while $Out_{T}^n(q,F_{A\!g})$ and $Out_{T}(q,F_{A\!g})$ consist of  finite and infinite traces starting from $q$  in which each step subjects to such choices.
We often omit the subscripts of $Out_{T}^n(q,F_{A\!g})$ and $Out_{T}(q,F_{A\!g})$  when the alternating transition system $T$ is clear from the context.

\begin{lemma}\label{lem:outcome+1}
For any $n\in \mathbb{N}$, the following conclusion holds:
\begin{equation*}
\begin{aligned} Out^{n+1}_T(q,F_{A\!g})=& \{s\in S^{n+1}: s[1,n]\in Out^n_T(q,F_{A\!g})\textrm{ and } \\ &F_{A\!g}(s[1,n])\cap Q=\{s[end]\}\textrm{ for some } Q\in\textrm{$\hbar$}(s[n],\Omega-Ag) \}.
\end{aligned}
\end{equation*}
\end{lemma}
\begin{proof}
Immediately.
\qed\end{proof}
\subsection{Alternating bisimilarity and alternating approximate bisimilarity}\label{subsec:pre_bisimilar}
To capture the behavioral equivalence between  alternating transition systems associated with the same observation set and agent set, Alur et al. introduce the notion of alternating bisimilarity~\cite{AlurCONCUr98}.

\begin{definition}\label{def:bisim}\cite{AlurCONCUr98}
Let $T_i=(S_i,\mathbb{P},\Omega,\Pi_i,\hbar_i)$ be two alternating transition systems $(i=1,2)$ and $Ag\subseteq \Omega$. The binary relation $R\subseteq S_1\times S_2$ is said to be an $Ag$-alternating bisimulation if and only if for any $(q_1,q_2)\in R$,

(1) $\Pi_1(q_1)=\Pi_2(q_2)$;

(2) $\forall Q_1\in\hbar_1(q_1,Ag) \exists Q_2\in\hbar_2(q_2,Ag)\forall Q'_2\in\hbar_2(q_2,\Omega-Ag)$$\exists Q'_1\in\hbar_1(q_1,\Omega-Ag) (Q_1\cap Q'_1)\times(Q_2\cap Q'_2)\subseteq R$
\footnote{By Definition~\ref{def:ATS}, it is easy to see that both  $Q_1\cap Q'_1$ and $Q_2\cap Q'_2$ are singleton. Therefore, $(Q_1\cap Q'_1)\times(Q_2\cap Q'_2)$ is single.};

(3) $\forall Q_2\in\hbar_2(q_2,Ag) \exists Q_1\in\hbar_1(q_1,Ag)\forall Q'_1\in\hbar_1(q_1,\Omega-Ag)$$\exists Q'_2\in\hbar_2(q_2,\Omega-Ag) (Q_1\cap Q'_1)\times(Q_2\cap Q'_2)\subseteq R$.

For any $q_1\in S_1$ and $q_2\in S_2$, these two states are said to be $Ag$-alternating bisimilar, in symbols $q_1\sim_{A\!g} q_2$, if and only if there exists an $Ag$-alternating bisimulation $R$ such that $(q_1,q_2)\in R$.
In other words, $\sim_{A\!g}\triangleq\bigcup\{R\subseteq S_1\times S_2: R\textrm{ is an Ag-alternating bisimulation}\}$.
\end{definition}

It is easy to check that $\sim_{A\!g}$ is an equivalence relation and is the largest $Ag$-alternating bisimulation.
We leave it to the interested readers.
Alur et al. establish a modal characterization of alternating bisimilarity  in terms of alternating-time temporal logic (for short, ATL)~\cite{AlurCONCUr98}.
They show that two states  are $Ag$-alternating bisimilar if and only if they satisfy the same $Ag$-ATL formulas, where  $Ag$-ATL formulas are ATL-formulas in which all path quantifiers occurring are parameterized by $Ag$.

Recently, alternating transition systems associated with metric over observations are adopted as models of the samples of control systems with disturbances and their finite abstractions~\cite{pola2008symbolic}\cite{pola2009symbolic}. In these work, the notion of alternating approximate bisimilarity is used to capture approximate equivalence between systems.

\begin{definition}\label{def:approximate bisimi}\cite{pola2009symbolic}
Let $T_i=(S_i,\mathbb{P},\Omega,\Pi_i,\hbar_i)$ be two alternating transition systems $(i=1,2)$ and $Ag\subseteq \Omega$. Suppose that ${d}$ is a metric over $\mathbb{P}$ and $\varepsilon\in \mathbb{R}^{0}_+$. The binary relation $R\subseteq S_1\times S_2$ is said to be an $(Ag,\varepsilon)$-alternating approximate bisimulation if and only if for any $(q_1,q_2)\in R$,

(1) ${d}(\Pi_1(q_1),\Pi_2(q_2))\leq \varepsilon$;

(2) $\forall Q_1\in\hbar_1(q_1,Ag) \exists Q_2\in\hbar_2(q_2,Ag)\forall Q'_2\in\hbar_2(q_2,\Omega-Ag)$$\exists Q'_1\in\hbar_1(q_1,\Omega-Ag) (Q_1\cap Q'_1)\times(Q_2\cap Q'_2)\subseteq R$;

(3) $\forall Q_2\in\hbar_2(q_2,Ag) \exists Q_1\in\hbar_1(q_1,Ag)\forall Q'_1\in\hbar_1(q_1,\Omega-Ag)$$\exists Q'_2\in\hbar_2(q_2,\Omega-Ag) (Q_1\cap Q'_1)\times(Q_2\cap Q'_2)\subseteq R$.

Two states $q_1\in S_1$ and $q_2\in S_2$ are said to be $(Ag,\varepsilon)$-alternating approximate bisimilar, denoted by $q_1\sim^{\varepsilon}_{A\!g}q_2$, if and only if there exists an $(Ag,\varepsilon)$-alternating approximate bisimulation $R\subseteq S_1\times S_2$ such that $(q_1,q_2)\in R$.

$T_1$ and $T_2$ are said to be $(Ag,\varepsilon)$-alternating approximate bisimilar, in symbols  $T_1\sim^{\varepsilon}_{A\!g}T_2$, if and only if $\{q_1\in S_1:q_1\sim^{\varepsilon}_{A\!g}q_2 \textrm{ for some } q_2\in S_2\}=S_1$  and $\{q_2\in S_2:q_1\sim^{\varepsilon}_{A\!g}q_2\textrm{ for some } q_1\in S_1\}=S_2$.
\end{definition}

The following results reveal some simple properties of $(Ag,\varepsilon)$-alternating approximate bisimilarity.

\begin{lemma}\label{lem:prop of Ag-E}
(1) $\sim^{0}_{A\!g}=\sim_{A\!g}$ and $\sim^{0}_{A\!g}$ is an equivalence relation;

(2) for any $\varepsilon_1,\varepsilon_2\in \mathbb{R}^0_+$, if $\varepsilon_1\leq \varepsilon_2$ then $\sim^{\varepsilon_1}_{A\!g}\subseteq \sim^{\varepsilon_2}_{A\!g}$;

(3) for any $\varepsilon\in \mathbb{R}^0_+$, $\sim^{\varepsilon}_{A\!g}$ is the largest $(Ag,\varepsilon)$-alternating approximate  bisimulation.
\end{lemma}
\begin{proof}
(1) Since ${d}$ is a metric, we have ${d}(\Pi_1(q_1), \Pi_2(q_2))\leq 0$ if and only if $\Pi_1(q_1)=\Pi_2(q_2)$  for any states $q_1$ and $q_2$. Thus it follows from Definition~\ref{def:bisim} and~\ref{def:approximate bisimi} that (1) holds.

(2) Immediately follows from Definition \ref{def:approximate bisimi}.

(3) Let $\varepsilon\in \mathbb{R}^0_+$. According to Definition \ref{def:approximate bisimi}, it is not difficult to check that $(Ag,\varepsilon)$-alternating approximate bisimulations are preserved under union.
Thus $\sim^{\varepsilon}_{A\!g}$ is the largest $(Ag,\varepsilon)$-alternating bisimulation.
\qed\end{proof}

As usual, $(Ag,\varepsilon)$-alternating approximate bisimilarity can be characterized in the forth-back style. Formally, we have

\begin{theorem}\label{th:approximate bisi}
$q_1\sim_{A\!g}^{\varepsilon} q_2$ if and only if the following hold:

(1) ${d}(\Pi_1(q_1),\Pi_2(q_2))\leq\varepsilon$;

(2) $\forall Q_1\in\hbar_1(q_1,Ag) \exists Q_2\in\hbar_2(q_2,Ag) \forall Q'_2\in\hbar_2(q_2,\Omega-Ag)$$\exists Q'_1\in\hbar_1(q_1,\Omega-Ag)((Q_1\cap Q'_1)\times (Q_2\cap Q'_2)\subseteq\sim^{\varepsilon}_{A\!g})$;

(3) $\forall Q_2\in\hbar_2(q_2,Ag) \exists Q_1\in\hbar_1(q_1,Ag) \forall Q'_1\in\hbar_1(q_1,\Omega-Ag)$$\exists Q'_2\in\hbar_2(q_2,\Omega-Ag)((Q_1\cap Q'_1)\times (Q_2\cap Q'_2)\subseteq\sim^{\varepsilon}_{A\!g})$.
\end{theorem}
\begin{proof}
(From left to right)
Follows from Definition~\ref{def:approximate bisimi} and (3) in Lemma~\ref{lem:prop of Ag-E}.

(From right to left)
Let $R\triangleq\{(q_1,q_2):q_1\textrm{ and } q_2 \textrm{ satisfy (1)-(3)}\}\cup\sim_{A\!g}^{\varepsilon}$. It is almost immediate  to check that $R$ is an $(Ag,\varepsilon)$-alternating approximate bisimulation.
So by (3) in Lemma~\ref{lem:prop of Ag-E}, the conclusion holds.
\qed\end{proof}

It should be pointed out that $(Ag,\varepsilon)$-alternating approximate bisimilarity is not always transitive and then is not always an equivalence relation.
An  example is given below.

\begin{example}\label{ex:not equivalent}
Consider the alternating transition system $(\{q_1,q_2,q_3\},\{p_1,p_2,p_3\},\\ \{1\},\Pi,\hbar)$, where $\Pi(q_i)=p_i$ and $\hbar(q_i,1)=\{\{q_i\}\}$ for $i=1,2,3$.
Let $Ag=\{1\}$.
Define a distance function ${d}$ over $\{p_1,p_2,p_3\}$ as: for any $p_i,p_j\in\{p_1,p_2,p_3\}$, ${d}(p_i,p_j)=|i-j|$.
Clearly, this function ${d}$ is a metric.
According to Definition~\ref{def:approximate bisimi}, it is not difficult to see that $p_1\sim^1_{A\!g}p_2$, $p_2\sim^1_{A\!g}p_3$ and $p_1\not\sim^1_{A\!g}p_3$.
Thus $\sim^1_{A\!g}$ is  not an equivalence relation.
\end{example}

\section{ATL$_\varepsilon$, $H^{\varepsilon}_{A\!g}$ and $E^{\varepsilon}_{A\!g}$}\label{sec:relation}

In this and the next sections, we will establish a logical characterization of $(Ag,\varepsilon)$-alternating approximate bisimilarity.
To this end, a modal language is introduced  below, which is obtained by adding the diamond operator $\langle\varepsilon\rangle$ to ATL.

\begin{definition}\label{def:logic app}
Let $\varepsilon\in \mathbb{R}^0_+$, $\mathbb{P}$ a finite set of propositions and let $\Omega$ be a set of agents.
ATL$_{\varepsilon}(\mathbb{P},\Omega)$ formulae are divided into: state formulas and path formulas, which are defined inductively as:
\begin{center}
state formula $\varphi::=p|\langle\varepsilon\rangle p|\neg\varphi|\varphi\wedge\varphi|\langle\!\langle Ag\rangle\!\rangle\phi,$
\end{center}
where $p\in \mathbb{P}$, $Ag\subseteq \Omega$ and $\phi$ is a path formula;
\begin{center}
path formula $\phi::=\varphi|\neg \phi|\phi\wedge\phi|X\phi|\phi\mathbf{U}\phi,$
\end{center}
where $\varphi$ is a state formula.

The operator $\langle\!\langle\ \rangle\!\rangle$ is a path quantifier.
Given $Ag\subseteq\Omega$, an ATL$_\varepsilon(\mathbb{P},\Omega)$ formula $\alpha$ is said to be an $Ag$-ATL$_\varepsilon(\mathbb{P},\Omega)$ formula if and only if all path quantifiers occurring in $\alpha$ are parameterized by $Ag$.
\end{definition}

As usual, logical connective $\vee$ can be defined in terms of $\neg$ and $\wedge$. If $\mathbb{P}$ and $\Omega$ are clear from the context, ATL$_\varepsilon(\mathbb{P},\Omega)$ and $Ag$-ATL$_\varepsilon(\mathbb{P},\Omega)$ are often abbreviated to  ATL$_\varepsilon$ and $Ag$-ATL$_\varepsilon$, respectively.
Henceforth, we use $\varphi,\gamma,\varphi_1,\gamma_1\cdots$ to denote state formulas and $\phi,\psi,\phi_1,\psi_1,\cdots$ to denote path formulas.

\begin{definition}\label{def:semantic of ATLe}
Let $T=(S,\mathbb{P},\Omega,\Pi,\hbar)$ be an alternating transition system, $d$ a metric over $\mathbb{P}$ and $\varepsilon\in\mathbb{R}_0^+$. The satisfaction relation $\models_{s}$ ($\models_p$) between the states (the infinite state sequence $\sigma\in S^\omega$, respectively) of $T$ and  state formulas (path formulas, respectively) is inductively defined as:  for any $q\in S$ and $\sigma\in S^{\omega}$,

$\bullet$ $(T,d),q\models_s p$ iff $p=\Pi(q)$ for any $p\in\mathbb{P}$;

$\bullet$ $(T,d),q\models_s\langle\varepsilon\rangle p$ iff ${d}(p,\Pi(q))\leq\varepsilon$;

$\bullet$ $(T,d),q\models_s\neg\varphi$ iff $T,q\models_s\varphi$ does not hold;

$\bullet$ $(T,d),q\models_s\varphi_1\wedge\varphi_2$ iff $(T,d),q\models_s\varphi_1$ and $(T,d),q\models_s\varphi_2$;

$\bullet$ $(T,d),q\models_s\langle\!\langle Ag\rangle\!\rangle\phi$ iff there exists a strategy  $F_{A\!g}$ of $Ag$ such that $(T,d),\sigma\models_p\phi$ for any $\sigma\in Out(q,F_{A\!g})$;

$\bullet$ for any state formula $\varphi$, $(T,d),\sigma\models_p \varphi$ iff $(T,d),\sigma[1]\models_s\varphi$;

$\bullet$ $(T,d),\sigma\models_p\mathbf{X}\phi$ iff $(T,d),\sigma[2,\infty]\models_p\phi$;

$\bullet$ $(T,d),\sigma\models_p\phi_1\mathbf{U}\phi_2$ iff there exists $i\in\mathbb{N}$ such that $(T,d),\sigma[i,\infty]\models_p\phi_2$ and for any $j<i$, $(T,d),\sigma[j,\infty]\models_p\phi_1$;

$\bullet$ $(T,d),\sigma\models_p\neg\phi$ (or $\phi_1\wedge\phi_2$) can be defined similarly to $\models_s$.
\end{definition}

For convenience, the subscripts of $\models_s$ and $\models_p$ will be omitted in this paper.
In the following, two rank functions are introduced as usual.

\begin{definition}\label{def:complexity of formula}
Let $\varepsilon\in\mathbb{R}_0^+$, $\mathbb{P}$ a finite set of propositions and let $\Omega$ be a set of agents.
The rank function $\xi_s$ ($\xi_p$) mapping ATL$_\varepsilon$ state formulas (path formulas, respectively) to natural numbers is defined as:

(1) for any $p\in\mathbb{P}$, $\xi_s(p)=1$ and $\xi_s(\langle\varepsilon\rangle p)=1$;

(2) $\xi_s(\neg \varphi)=\xi_s(\varphi)+1$;

(3) $\xi_s(\langle\!\langle Ag\rangle\!\rangle\phi)=\xi_p(\phi)+1$;

(4) $\xi_s(\varphi_1\wedge\varphi_2)=\max\{\xi_s(\varphi_1),\xi_s(\varphi_2)\}+1$;

(5) for any state formula $\varphi$, $\xi_p(\varphi)=\xi_s(\varphi)+1$;

(6) $\xi_p(\neg\phi)=\xi_p(\phi)+1$;

(7) $\xi_p(\mathbf{X}\phi)=\xi_p(\phi)+1$;

(8) $\xi_p(\phi_1\wedge\phi_2)=\max\{\xi_p(\phi_1),\xi_p(\phi_2)\}+1$;

(9) $\xi_p(\phi_1\mathbf{U}\phi_2)=\max\{\xi_p(\phi_1),\xi_p(\phi_2)\}+1$.
\end{definition}

This paper aims to establish a modal characterization of $(Ag,\varepsilon)$-alternating approximate bisimilarity in terms of $Ag$-ATL$_{\varepsilon}$.
However, as shown in Example~\ref{ex:not equivalent}, $(Ag,\varepsilon)$-alternating approximate bisimilarity is not always an equivalence relation. Then it may not coincide with modal equivalence w.r.t any modal logic.
In other words, the modal characterization of $(Ag,\varepsilon)$-alternating approximate bisimilarity can not be provided in the usual style.

To overcome this defect, two binary relations between formulas will be introduced, which will  play the central roles in this paper.
Before giving them formally, we explain the motivation behind these
notions. Recall that two states are $(Ag,\varepsilon)$-alternating approximate bisimilar if and only if they
satisfy the forth and back conditions in Theorem~\ref{th:approximate bisi}. So, in order to
establish the modal characterization of $(Ag,\varepsilon)$-alternating approximate bisimilarity, we need to
formalize these conditions in terms of ATL$_{\varepsilon}$ formulas. According to the semantics of ATL$_{\varepsilon}$, we have the following observation.

\textit{For any $\varepsilon\in \mathbb{R}_0^+$, state $q_1$ of $T_1$ and state $q_2$ of $T_2$, $q_1\sim^\varepsilon_{A\!g} q_2$ implies that for each ~$p\in\mathbb{P}$, $(T_1,d), q_1\models p$ implies $(T_2,d), q_2\models
\left\langle \varepsilon \right\rangle p$ and vice verse.}

This simple observation gives us a hint about the logical characterization
of $(Ag,\varepsilon)$-alternating approximate bisimilarity. That is, we may characterize it in terms of an appropriate binary relation $H$ over
ATL$_\varepsilon$ state formulae, and this
characterization will possess the form ``$q_1\sim^\varepsilon_{A\!g} q_2$ iff for any pair
$(\varphi $, $\gamma )\in H$, $(T_1,d), q_1\models \varphi $ implies $%
(T_2,d),q_2\models \gamma $, and vice versa''. To provide such relation~$H$, we introduce the notions below.

\begin{definition}\label{def:relation of formula 1}
Let $\mathbb{P}$ be a finite set of propositions, $\varepsilon \in \mathbb{R}^0_+$, $\Omega$ a set of agents and $Ag\subseteq \Omega$.
The binary relation $H^{\varepsilon}_{A\!g}(\mathbb{P},\Omega)$ over $Ag$-ATL$_{\varepsilon}$ state formulas and the binary relation $E^{\varepsilon}_{A\!g}(\mathbb{P},\Omega)$ over $Ag$-ATL$_{\varepsilon}$ path formulas are the smallest pair of relations satisfying the following conditions (i.e., for any pair of relations $H$ and $E$ over states formulas and path formulas, respectively, if they satisfy the following conditions then $H^{\varepsilon}_{A\!g}(\mathbb{P},\Omega)\subseteq H$ and $E^{\varepsilon}_{A\!g}(\mathbb{P},\Omega)\subseteq E$):

(1) for any $p\in \mathbb{P}$, $(p,\langle\varepsilon\rangle p) \in H^{\varepsilon}_{A\!g}(\mathbb{P},\Omega)$;

(2) if $(\varphi, \gamma)\in H^{\varepsilon}_{A\!g}(\mathbb{P},\Omega) $, then $(\neg\gamma, \neg\varphi)\in H^{\varepsilon}_{A\!g}(\mathbb{P},\Omega) $;

(3) if $(\varphi_i, \gamma_i)\in H^{\varepsilon}_{A\!g}(\mathbb{P},\Omega) $ for $i=1,2$, then $(\varphi_1\wedge\varphi_2,\gamma_1\wedge\gamma_2)\in H^{\varepsilon}_{A\!g}(\mathbb{P},\Omega) $;

(4) if $(\psi,\phi)\in E^{\varepsilon}_{A\!g}(\mathbb{P},\Omega) $, then $(\langle\!\langle Ag\rangle\!\rangle\psi, \langle\!\langle Ag\rangle\!\rangle \phi)\in H^{\varepsilon}_{A\!g}(\mathbb{P},\Omega) $;

(5) if $(\varphi, \gamma)\in H^{\varepsilon}_{A\!g}(\mathbb{P},\Omega)$, then $(\varphi, \gamma)\in E^{\varepsilon}_{A\!g}(\mathbb{P},\Omega) $;

(6) if  $(\psi,\phi)\in E^{\varepsilon}_{A\!g}(\mathbb{P},\Omega) $, then $(\neg\phi,\neg\psi)\in E^{\varepsilon}_{A\!g}(\mathbb{P},\Omega) $;

(7) if $(\psi_i,\phi_i)\in E^{\varepsilon}_{A\!g}(\mathbb{P},\Omega)$ for $i=1,2$, then $(\psi_1\wedge\psi_2,\phi_1\wedge\phi_2)\in E^{\varepsilon}_{A\!g}(\mathbb{P},\Omega)$;

(8) if $(\psi,\phi)\in E^{\varepsilon}_{A\!g}(\mathbb{P},\Omega)$, then $(\mathbf{X}\psi,\mathbf{X}\phi)\in E^{\varepsilon}_{A\!g}(\mathbb{P},\Omega) $;

(9) if $(\psi_i,\phi_i)\in E^{\varepsilon}_{A\!g}(\mathbb{P},\Omega)$ for $i=1,2$, then $(\psi_1\mathbf{U}\psi_2, \phi_1\mathbf{U}\phi_2)\in E^{\varepsilon}_{A\!g}(\mathbb{P},\Omega)$.
\end{definition}

For convenience,  if $\mathbb{P}$ and $\Omega$ are clear from the context, $H^{\varepsilon}_{A\!g}(\mathbb{P},\Omega)$ and $E^{\varepsilon}_{A\!g}(\mathbb{P},\Omega)$ are often abbreviated to $H^{\varepsilon}_{A\!g}$ and $E^{\varepsilon}_{A\!g}$, respectively.
The following result guarantees the existence of these two relations.

\begin{proposition}\label{pro:preserve under intersection}
Let $\mathbb{P}$ be a finite set of propositions, $\varepsilon \in \mathbb{R}^0_+$, $\Omega$ a set of agents and $Ag\subseteq \Omega$. Then

(i) Let $I$ be an index set. If for each $i\in I$, the binary relation $H_i$ over $Ag$-ATL$_{\varepsilon}$ state formulas and the binary relation $E_i$ over $Ag$-ATL$_{\varepsilon}$ path formulas satisfy the conditions in Definition~\ref{def:relation of formula 1}, then so is $(\bigcap_{i\in I}H_i,\bigcap_{i\in I}E_i)$.

(ii) The smallest pair of relations satisfying the conditions in Definition~\ref{def:relation of formula 1} exist.
\end{proposition}
\begin{proof}
Clearly, the pair of relations $H$ and $E$ satisfy the conditions in Definition~\ref{def:relation of formula 1}, where $H\triangleq\{(\varphi_1,\varphi_2): \varphi_1$ and $\varphi_2$ are $Ag$-ATL$_{\varepsilon}$ state formulas\} and $E\triangleq\{(\phi_1,\phi_2): \phi_1$ and $\phi_2$ are $Ag$-ATL$_{\varepsilon}$ path formulas$\}$.
So it follows from (i) that (ii) holds.
We prove (i) below.

Assume that for each $i\in I$, the binary relation $H_i$ over $Ag$-ATL$_{\varepsilon}$ state formulas and the binary relation $E_i$ over $Ag$-ATL$_{\varepsilon}$ path formulas satisfy the conditions in Definition~\ref{def:relation of formula 1}.
It suffices to show that the pair $(\bigcap_{i\in I}H_i,\bigcap_{i\in I}E_i)$ satisfies the conditions $(1)$-(9)  in Definition~\ref{def:relation of formula 1}. We will provide two sample cases.

(1) Let $p\in \mathbb{P}$. Then for each $i\in I$, since $H_i$  and  $E_i$ satisfy the conditions in Definition~\ref{def:relation of formula 1}, it follows that $(p,\langle\varepsilon \rangle p)\in H_i$. Thus we have $(p,\langle\varepsilon \rangle p)\in \bigcap_{i\in I} H_i$.

(2) Let $(\varphi, \gamma)\in \bigcap_{i\in I} H_i$. Then $(\varphi, \gamma)\in H_i$ for each $i\in I$. So for each $i\in I$, since $H_i$  and  $E_i$ satisfy the conditions in Definition~\ref{def:relation of formula 1}, we get $(\neg\gamma,\neg\varphi)\in H_i$. Therefore, we obtain $(\neg\gamma,\neg\varphi)\in \bigcap_{i\in I} H_i$.
\qed\end{proof}

A few of useful properties of $H^{\varepsilon}_{A\!g}$ and $E^{\varepsilon}_{A\!g}$ are listed below.

\begin{lemma}\label{lem:properties of H and E}
For any $(\varphi, \gamma)\in H^{\varepsilon}_{A\!g}$ and $(\psi, \phi)\in E^{\varepsilon}_{A\!g}$, the following hold:

$(a)$ $\varphi$ is in one of the following forms: $p$, $\neg\gamma_1, \varphi_1\wedge\varphi_2$ and $\langle\!\langle Ag\rangle\!\rangle\psi$;

$(b)$ $\varphi$ and $\psi$ can not be in the form of $\langle\varepsilon\rangle p$;

$(c)$ if $\varphi=p$, then $\gamma=\langle\varepsilon\rangle p$ and $\xi_s(\varphi)=\xi_s(\gamma)=1$;

$(d)$ if $\varphi=\neg \gamma_1$, then there exists a state formula $\varphi_1$ such that $(\varphi_1,\gamma_1)\in H^{\varepsilon}_{A\!g}$, $\gamma=\neg\varphi_1$ and $\xi_{s}(\varphi)=\xi_s(\gamma)=\xi_s(\gamma_1)+1$;

$(e)$ if $\varphi=\varphi_1\wedge\varphi_2$, then there exist state formulas $\gamma_1$ and $\gamma_2$ such that $(\varphi_i,\gamma_i)\in H^{\varepsilon}_{A\!g}$ $(i=1,2)$, $\gamma=\gamma_1\wedge\gamma_2$ and $\xi_s(\varphi)=\xi_s(\gamma)=\max\{\xi_s(\gamma_1),\xi_s(\gamma_2)\}+1$;

$(f)$ if $\varphi=\langle\!\langle Ag\rangle\!\rangle\psi_1$, then there exists a path formula $\phi_1$ such that $(\psi_1,\phi_1)\in E^{\varepsilon}_{A\!g}$, $\gamma=\langle\!\langle Ag\rangle\!\rangle \phi_1$ and $\xi_s(\varphi)=\xi_s(\gamma)=\xi_p(\psi_1)+1$;

$(g)$ if the path formula $\psi$ is also a state formula, then $\phi$ is also a state formula, $(\psi,\phi)\in H^{\varepsilon}_{A\!g}$ and $\xi_p(\phi)=\xi_p(\psi)=\xi_s(\psi)+1$;

$(h)$ if $\psi=\neg\phi_1$, then there exists a path formula $\psi_1$ such that $(\psi_1,\phi_1)\in E^{\varepsilon}_{A\!g}$, $\phi=\neg\psi_1$ and $\xi_p(\phi)=\xi_p(\psi)=\xi_p(\psi_1)+1$;

$(i)$ if  $\psi=\psi_1\wedge\psi_2$, then there exist path formulas $\phi_1$ and $\phi_2$ such that $(\psi_i,\phi_i)\in E^{\varepsilon}_{A\!g} (i=1,2)$, $\phi=\phi_1\wedge\phi_2$ and $\xi_p(\phi)=\xi_p(\psi)=\max\{\xi_p(\psi_1),\xi_p(\psi_2)\}+1$;

$(j)$ if $\psi=\mathbf{X}\psi_1$, then there exists a path formula $\phi_1$ such that $(\psi_1,\phi_1)\in E^{\varepsilon}_{A\!g}$, $\phi=\mathbf{X}\phi_1$ and $\xi_p(\phi)=\xi_p(\psi)=\xi_p(\psi_1)+1$;

$(k)$ if $\psi=\psi_1\mathbf{U}\psi_2$, then there exist path formulas $\phi_1$ and $\phi_2$ such that $(\psi_i,\phi_i)\in E^{\varepsilon}_{A\!g} (i=1,2)$, $\phi=\phi_1\mathbf{U}\phi_2$ and $\xi_p(\phi)=\xi_p(\psi)=\max\{\xi_p(\psi_1),\xi_p(\psi_2)\}+1$.
\end{lemma}
\begin{proof}
Straightforward.
\qed\end{proof}

It follows from the above result that $(p,\langle\varepsilon\rangle p)\in H^{\varepsilon}_{A\!g}$ and $(p,\langle\varepsilon\rangle p)\in E^{\varepsilon}_{A\!g}$ but neither $(\langle\varepsilon\rangle p, p)\in H^{\varepsilon}_{A\!g}$ nor $(\langle\varepsilon\rangle p,p)\in E^{\varepsilon}_{A\!g}$.
Thus neither $H^{\varepsilon}_{A\!g}$ nor $E^{\varepsilon}_{A\!g}$ is symmetric and then none of them is  an equivalence relation.
For $\varepsilon=0$, some internal relations between $H^\varepsilon_{A\!g}$ (or, $E^\varepsilon_{A\!g}$) and $\models_s$ ($\models_p$, respectively) are revealed in the next lemma.

\begin{lemma}\label{lem:basic properties of H and E}
Let $T=(S,\mathbb{P},\Omega,\Pi,\hbar)$ be an alternating transition system and let ${d}$ be a metric over $\mathbb{P}$.
Then

(1) for any $q\in S$ and $(\varphi,\gamma)\in H^{0}_{A\!g}$, $(T,d),q\models\varphi$ if and only if $(T,d),q\models\gamma$;

(2) for any $\sigma\in S^{\omega}$ and $(\psi, \phi)\in E^0_{A\!g}$, $(T,d),\sigma \models\psi$ if and only if $(T,d),\sigma \models\phi$;

(3)for any $Ag$-ATL$_0$ state formula $\varphi_0$, there exists $(\varphi,\gamma)\in H^0_{A\!g}$ such that for any $q\in S$, $(T,d),q\models\varphi_0$ if and only if $(T,d),q\models\varphi$;

(4) for any $Ag$-ATL$_0$ path formula $\psi_0$, there exists $(\psi,\phi)\in E^0_{A\!g}$ such that for any $\sigma\in S^{\omega}$, $(T,d),\sigma\models\psi_0$ if and only if $(T,d),\sigma\models\psi$.
\end{lemma}
\begin{proof}
Since ${d}$ is a metric over $\mathbb{P}$, for any $p,p'\in \mathbb{P}$, ${d}(p,p')\leq 0$ if and only if $p=p'$.
Further, by Definition~\ref{def:semantic of ATLe} and Lemma~\ref{lem:properties of H and E}, it is easy to prove (1) and (2) by induction on the ranks of $\varphi$ and $\psi$.
Next, we prove (3) and (4) simultaneously by induction  on the ranks of $\varphi_0$ and $\psi_0$.

By Definition~\ref{def:complexity of formula}, it is clear that if $\xi_s(\varphi_0)=0$ and $\xi_p(\psi_0)=0$ then (3) and (4) hold.

Suppose that $\xi_s(\varphi_0)=\xi_p(\psi_0)=n+1$ and the items (3) and (4) hold for any $Ag$-ATL$_0$ state formula $\varphi$ and $Ag$-ATL$_0$ path formula $\psi$ with $\xi_s(\varphi)\leq n$ and $\xi_p(\psi)\leq n$.
According to Definition~\ref{def:complexity of formula}, $\varphi_0$ is in one of the following forms: $p$, $\langle 0 \rangle p$, $\neg\gamma_1, \varphi_1\wedge\varphi_2$ and $\langle\!\langle Ag\rangle\!\rangle\psi$, and $\psi_0$ is in the form of  $\varphi$, $\neg\phi_1, \psi_1\wedge\psi_2$, $\mathbf{X}\psi$ or $\psi_1\mathbf{U}\psi_2$.
In the following, we just provide two sample cases. The proofs of other cases are similar.

Suppose that $\varphi_0=\langle 0\rangle p$ for some $p\in \mathbb{P}$.
It follows from (1) and $(p,\langle 0\rangle p)\in H^0_{A\!g}$ that for any $q\in S$, $(T,d),q\models \varphi_0$ if and only if $(T,d),q\models p$.
We set $\varphi\triangleq p$ and $\gamma\triangleq \langle 0\rangle p$.
Clearly, $(\varphi,\gamma)\in H^0_{A\!g}$ and for any $q\in S$, $(T,d),q\models\varphi_0$ if and only if $(T,d),q\models\varphi$.

Suppose that $\varphi_0=\langle\!\langle Ag\rangle\!\rangle\psi$. Then by Definition~\ref{def:complexity of formula}, we get $\xi_p(\psi)=n$. Further, by induction hypothesis, there exists $(\psi',\phi')\in E^0_{A\!g}$ such that for any $\sigma\in S^\omega$, $(T,d),\sigma\models\psi$ if and only if $(T,d),\sigma\models\psi'$.
We set $\varphi\triangleq \langle\!\langle Ag\rangle\!\rangle\psi'$ and $\gamma\triangleq \langle\!\langle Ag\rangle\!\rangle\phi'$.
Then it follows from Definition~\ref{def:relation of formula 1} that $(\varphi,\gamma)\in H^0_{A\!g}$.
Moreover, by Definition~\ref{def:semantic of ATLe}, it is clear that for any $q\in S$, $(T,d),q\models\varphi_0$ if and only if $(T,d),q\models\varphi$.
\qed\end{proof}

\begin{proposition}\label{pro:basic properties of H and E}
Let $T=(S,\mathbb{P},\Omega,\Pi,\hbar)$ be an alternating transition system, $d$ a metric over $\mathbb{P}$ and $\varepsilon\in \mathbb{R}^0_+$.
Then for any $q\in S$, $\sigma\in S^{\omega}$, $(\varphi,\gamma)\in H^{\varepsilon}_{A\!g}$ and $(\psi,\phi)\in E^{\varepsilon}_{A\!g}$, the following conclusions hold:

(1) if $(T,d),q\models\varphi$ then $(T,d),q\models\gamma$;

(2) if $(T,d),\sigma\models\psi$ then $(T,d),\sigma\models\phi$.
\end{proposition}
\begin{proof}
We prove (1) and (2) simultaneously by induction on the ranks of $\varphi$ and $\psi$.

By Definition~\ref{def:complexity of formula}, if $\xi_s(\varphi)=0$ and $\xi_p(\psi)=0$ then the conclusions hold trivially.

Suppose that $\xi_s(\varphi)=\xi_p(\psi)=n+1$, $(\varphi,\gamma)\in H^{\varepsilon}_{A\!g}$, $(\psi,\phi)\in E^{\varepsilon}_{A\!g}$ and the items (1) and (2) hold for any $(\varphi_0,\gamma_0)\in H^{\varepsilon}_{A\!g}$ and $(\psi_0,\phi_0)\in E^{\varepsilon}_{A\!g}$ with $\xi_s(\varphi_0)\leq n$ and $\xi_p(\psi_0)\leq n$.
By Definition~\ref{def:logic app}, $\varphi$ is in one of the following forms: $p$, $\neg\gamma_1, \varphi_1\wedge\varphi_2$ and $\langle\!\langle Ag\rangle\!\rangle\psi_1$, and $\psi$ is in the form of $\varphi_1$, $\neg\phi_1, \psi_1\wedge\psi_2$, $\mathbf{X}\psi_1$ or $\psi_1\mathbf{U}\psi_2$.
In the following, we just provide some sample cases.

\textbf{Case 1} $\varphi=p$ for some $p\in\mathbb{P}$.
Then by Lemma~\ref{lem:properties of H and E}, we obtain $\gamma=\langle\varepsilon\rangle p$.
Let $q\in S$ and $(T,d),q\models \varphi$.
Then it follows from  Definition~\ref{def:semantic of ATLe} that $\Pi(q)=p$.
Since $d$ is a metric, we have $d(p,\Pi(q))=d(p,p)=0\leq \varepsilon$.
Thus by Definition~\ref{def:semantic of ATLe}, we get $(T,d),q\models \langle \varepsilon\rangle p$.

\textbf{Case 2} $\varphi=\langle \langle Ag \rangle\!\rangle\psi_0$.
Let $q\in S$ and $(T,d),q\models\varphi$.
Due to Lemma~\ref{lem:properties of H and E}, there exists a path formula $\phi_0$ such that $(\psi_0,\phi_0)\in E^{\varepsilon}_{A\!g}$, $\gamma=\langle\!\langle Ag\rangle\!\rangle \phi_0$ and $\xi_s(\varphi)=\xi_s(\gamma)=\xi_p(\psi_0)+1$.
By Definition~\ref{def:semantic of ATLe} and $(T,d),q\models\varphi$, there is a strategy $F_{A\!g}$ of $Ag$ such that $(T,d),\sigma\models\psi_0$ for any $\sigma\in Out(q,F_{A\!g})$.
Then by induction hypothesis, $(T,d),\sigma\models\phi_0$ for any $\sigma\in Out(q,F_{A\!g})$.
Therefore, by Definition~\ref{def:semantic of ATLe}, we have $(T,d),q\models\gamma$.

\textbf{Case 3} $\psi=\psi_1\mathbf{U}\psi_2$.
Let $\sigma\in S^{\omega}$ and $(T,d),\sigma\models\psi$.
It follows from Lemma~\ref{lem:properties of H and E} that for some path formulas $\phi_1$ and $\phi_2$, $(\psi_i,\phi_i)\in E^{\varepsilon}_{A\!g}$ $(i=1,2)$, $\phi=\phi_1\mathbf{U}\phi_2$ and $\xi_p(\psi)=\xi_p(\phi)=\max\{\xi_p(\psi_1),\xi_p(\psi_2)\}+1$.
Since $(T,d),\sigma\models\psi$, by Definition~\ref{def:semantic of ATLe}, there exists $i\in\mathbb{N}$ such that $(T,d),\sigma[i,\infty]\models\psi_2$ and  $(T,d),\sigma[j,\infty]\models\psi_1$ for any $j<i$.
Then by induction hypothesis, $(T,d),\sigma[i,\infty]\models\phi_2$ and  $(T,d),\sigma[j,\infty]\models\phi_1$ for any $j<i$.
Therefore, by Definition~\ref{def:semantic of ATLe}, we obtain $(T,d),\sigma\models\phi$.
\qed\end{proof}

\section{Modal characterization of alternating approximate  bisimilarity}\label{sec:characterization}
This section will establish a modal characterization of $(Ag,\varepsilon)$-alternating approximate bisimilarity in terms of relations $H^{\varepsilon}_{A\!g}$ and $E^{\varepsilon}_{A\!g}$ defined in the previous section.
Similar method has been adopted to provide the modal characterization of $\lambda$-bisimilarity~\cite{zhang2007modal}.
In order to obtain such modal characterization, a number of auxiliary lemmas are needed.

Firstly, we intend to demonstrate that for any alternating approximately bisimilar states $q_1$ (of $T_1$) and $q_2$ (of $T_2$), given a strategy of $T_1$, there exists a strategy of $T_2$ such that, under control of these strategies, each trace starting from $q_2$ is approximately bisimilar to some trace starting from $q_1$.
To prove this conclusion, we need k$\ddot{o}$nig's lemma (see~\cite{konig1936}), which says
\begin{center}
every infinite, finite branching tree has an infinite branch.
\end{center}

\begin{lemma}\label{lem:characterization 1}
Let $T_i=(S_i,\mathbb{P},\Omega,\Pi_i,\hbar_i)$ be two finite branching alternating transition systems $(i=1,2)$.
Suppose that $d$ is a metric over $\mathbb{P}$, $\varepsilon\in\mathbb{R}^0_+$,
$Ag\subseteq\Omega$ and $F_{A\!g}:(S_1)^+\rightarrow 2^{S_1}$ is a strategy of $Ag$.
For any $q_1\in S_1$ and $q_2\in S_2$ with $q_1\sim_{A\!g}^\varepsilon q_2$, there is a strategy $F'_{A\!g}:(S_2)^+\rightarrow 2^{S_2}$ such that for any $\sigma_2\in Out(q_2,F'_{A\!g})$, $\sigma_1\sim^{\varepsilon}_{A\!g}\sigma_2$ for some $\sigma_1\in Out(q_1,F_{A\!g})$ \footnote{$\sigma_1\sim^{\varepsilon}_{A\!g}\sigma_2$ if and only if for any $i\in \mathbb{N}$, $\sigma_1[i]\sim^{\varepsilon}_{A\!g}\sigma_2[i]$. Similarly, $s_1\sim^{\varepsilon}_{A\!g} s_2$ if and only if $s_1[i]\sim^{\varepsilon}_{A\!g} s_2[i]$ for any $i\leq \max\{|s_1|,|s_2|\}$.}.
\end{lemma}
\begin{proof}
Let $q_1\in S_1$, $q_2\in S_2$ and $q_1\sim_{A\!g}^\varepsilon q_2$. To obtain the desired strategy $F'_{A\!g}:(S_2)^+\rightarrow 2^{S_2}$, we define subsets $\Delta_n$ of $(S_2)^n$ and functions $F_n:\Delta_n\rightarrow 2^{S_2} (n\in\mathbb{N})$ by induction on $n$ as follows.

We set $\Delta_1\triangleq \{q_2\}$.
Since $F_{A\!g}$ is a strategy of $Ag$, we get $F_{A\!g}(q_1)\in\hbar_1(q_1,Ag)$.
Then by $q_1\sim_{A\!g}^\varepsilon q_2$ and Theorem~\ref{th:approximate bisi}, there exists $Q_2\in\hbar_2(q_2,Ag)$ such that for any $Q'_2\in\hbar_2(q_2,\Omega-Ag)$, $(Q'_1\cap F_{A\!g}(q_1))\times (Q_2\cap Q'_2)\subseteq \sim_{A\!g}^\varepsilon$ for some $Q'_1\in\hbar_1(q_1,\Omega-Ag)$.
Note that such $Q_2$ may not be unique.
Choose  and fix an arbitrary such $Q_2$ and set $F_1(q_2)\triangleq Q_2$.
Clearly, $F_1$ is a function from $\Delta_1$ to $2^{S_2}$.

Suppose that $\Delta_k$ and $F_k$ have been defined. We define $\Delta_{k+1}$ and $F_{k+1}$ below.
We set
\begin{eqnarray*}
\Delta_{k+1}\triangleq\{s_2 q'_2 :s_2 \in \Delta_k \textrm{ and } q'_2\in F_k(s_2)\cap Q'_2 \textrm{ for some } Q'_2\in \hbar_2(s_2[end],\Omega-Ag)\}.
\end{eqnarray*}

For any $s'_2\in\Delta_{k+1}$, if there does not exist $s'_1\in Out^{k+1}(q_1,F_{A\!g})$ such that $s'_1\sim^\varepsilon_{A\!g}s'_2$, then we set $F_{k+1}(s'_2)=S_2$; if there exists $s'_1\in Out^{k+1}(q_1,F_{A\!g})$ such that $s'_1\sim^\varepsilon_{A\!g}s'_2$, then by Theorem~\ref{th:approximate bisi}, we have

\begin{equation*}
\begin{aligned}
\exists Q_2\in\hbar_2(s'_2[end],Ag)(&\forall Q'_2\in\hbar_2(s'_2[end],\Omega-Ag)\exists Q'_1\in\hbar_1(s'_1[end],\Omega-Ag)
\\
&((Q'_1\cap F_{A\!g}(s'_1))\times (Q_2\cap Q'_2)\subseteq \sim^{\varepsilon}_{A\!g})).
\end{aligned}
\end{equation*}

We choose such a $Q_2$ and set $F_{k+1}(s'_2)\triangleq Q_2$.
Clearly, $F_{k+1}$  is a function from $\Delta_{k+1}$ to $2^{S_2}$.

On the other hand, by Lemma~\ref{lem:stra_Ag}, there exists at least one strategy of $Ag$ mapping $(S_2)^+$ to $2^{S_2}$.
Let $F''_{A\!g}:(S_2)^+\rightarrow 2^{S_2}$ be an arbitrary strategy of $Ag$.
Define the function $F'_{A\!g}:(S_2)^+\rightarrow 2^{S_2}$ as follows: for any $s\in(S_2)^+$,
\begin{equation*}
F'_{A\!g}(s)=\left\{\begin{aligned}& F_{|s|}(s)   &\textrm{  if $s\in\Delta_{|s|}$}\\
& F''_{A\!g}(s)   &\textrm{otherwise}
\end{aligned}\right.
\end{equation*}

Next we want to show that the function $F'_{A\!g}$ is the desired strategy. It is enough to prove three claims below.

\textbf{Claim 1}
For any $n\in\mathbb{N}$, the following conclusions hold:

$(1_n)$ $\Delta_n\neq \emptyset$;

$(2_n)$ for any $s_2\in \Delta_n$, $s_1\sim^\varepsilon_{A\!g}s_2$ for some $s_1\in Out^n(q_1,F_{A\!g})$;

$(3_n)$ for any $s_2\in \Delta_n$, $F_n(s_2)\in \hbar_2(s_2[end],Ag)$.

We proceed by induction on $n$.

If $n=1$, then $(1_n)$ and $(2_n)$ hold trivially. By the definition of $F_1$, we obtain $F_1(q_2)\in\hbar_2(q_2, Ag)$. Thus ($3_n$) holds due to $\Delta_1=\{q_2\}$.

Suppose that $(1_k)$, $(2_k)$ and $(3_k)$ hold. We prove $(1_{k+1})$, $(2_{k+1})$ and $(3_{k+1})$ in turn.

$(1_{k+1})$ By induction hypothesis, we get $\Delta_k\neq \emptyset$ and $F_k(s)\in\hbar_2(s[end],Ag)$ for any $s\in\Delta_k$.
Let $s_2\in\Delta_k$ and $Q'_2\in \hbar_2(s_2[end],\Omega-Ag)$.
By Definition~\ref{def:ATS}, $F_k(s_2)\cap Q'_2$ is a singleton set.
That is, there exists $q'_2\in S_2$ such that $F_k(s_2)\cap Q'_2 =\{q'_2\}$.
Therefore, it follows from the definition of $\Delta_{k+1}$ that $s_2 q'_2\in\Delta_{k+1}$.
Thus $\Delta_{k+1}\neq \emptyset$.

$(2_{k+1})$ Let $s'_2\in \Delta_{k+1}$. By the construction of $\Delta_{k+1}$, there exists $s_2\in\Delta_k$, $Q'_2\in\hbar(s_2[end],\Omega-Ag)$ and $q'_2\in F_k(s_2)\cap Q'_2$ such that $s'_2 =s_2 q'_2$.
By induction hypothesis,  $s_1\sim^\varepsilon_{A\!g} s_2$ for some $s_1\in Out^{k}(q_1,F_{A\!g})$.
Then by the construction of $F_k$, there exists $Q'_1\in\hbar_1(s_1[end],\Omega-Ag)$ such that
\begin{equation*}
(Q'_1\cap F_{A\!g}(s_1))\times (F_k(s_2)\cap Q'_2)\subseteq \sim^\varepsilon_{A\!g}.
\end{equation*}
According to Definition~\ref{def:ATS}, both $Q'_1\cap F_{A\!g}(s_1)$ and $F_k(s_2)\cap Q'_2$ are singleton sets.
Thus there exists $q'_1\in S_1$ such that $Q'_1\cap F_{A\!g}(s_1)=\{q'_1\}$ and $q'_1\sim^\varepsilon_{A\!g} q'_2$.
Then it follows that $s_1q'_1\sim^\varepsilon_{A\!g} s'_2$.
Moreover, by Lemma~\ref{lem:outcome+1}, we have $s_1 q'_1\in Out^{k+1}(q_1,F_{A\!g})$, as desired.

$(3_{k+1})$ Let $s'_2\in \Delta_{k+1}$. It follows from $(2_{k+1})$ that $s'_1\sim^\varepsilon_{A\!g}s'_2$ for some $s'_1\in Out^{k+1}(q_1,F_{A\!g})$.
Then by the definition of $F_{k+1}$, we get $F_{k+1}(s'_2)\in\hbar_2(s'_2[end],Ag)$.
\\

\textbf{Claim 2}
$F'_{A\!g}$ is a strategy of $Ag$ and for any $n\in\mathbb{N}$, $\Delta_n=Out^{n}(q_2,F'_{A\!g})$.

By Claim 1 and the definition of $F'_{A\!g}$, $F'_{A\!g}(s_2)\in \hbar_2(s_2[end], Ag)$ for any $s_2\in (S_2)^+$.
Thus by Lemma~\ref{lem:stra_Ag}, $F'_{A\!g}$ is a strategy of $Ag$. Next we prove that for any $n\in\mathbb{N}$, $\Delta_n=Out^{n}(q_2,F'_{A\!g})$.

If $n=1$ then $\Delta_n=\{q_2\}=Out^n(q_2,F'_{A\!g})$ holds trivially.

Suppose that $n=k+1$ and $\Delta_k=Out^k(q_2, F'_{A\!g})$.
Then by the definition of $F'_{A\!g}$, $F'_{A\!g}(s_2)=F_k(s_2)$ for any $s_2\in\Delta_k$.
Further, since $\Delta_k=Out^k(q_2,F'_{A\!g})$, by the definition of $\Delta_{k+1}$ and Lemma~\ref{lem:outcome+1}, it is easy to check that $\Delta_{k+1}=Out^{k+1}(q_2,F'_{A\!g})$.
\\

\textbf{Claim 3}
For any $\sigma_2\in Out(q_2,F'_{A\!g})$, there exists $\sigma_1\in Out(q_1,F_{A\!g})$ such that $\sigma_1\sim^\varepsilon_{A\!g} \sigma_2$.

Suppose that $\sigma_2\in Out(q_2,F'_{A\!g})$.
By Definition~\ref{def:outcome} and Claim 2, for any $n\in\mathbb{N}$, $\sigma_2[1,n]\in Out^{n}(q_2,F'_{A\!g})=\Delta_{n}$.
In order to demonstrate the existence of the desired $\sigma_1\in Out(q_1,F_{A\!g})$, we construct the tree
\begin{center}
$<\bigcup_{n\in\mathbb{N}}\{s_n\in Out^n(q_1,F_{A\!g}):s_n\sim_{A\!g}^\varepsilon \sigma_2[1,n]\}, R,q_1>$
\end{center}
 as:

$(a)$ $q_1$ is the root;

$(b)$ for any $s,s'\in \bigcup_{n\in\mathbb{N}}\{s_n\in Out^n(q_1,F_{A\!g}):s_n\sim_{A\!g}^\varepsilon \sigma_2[1,n]\}$, $sRs'$ if and only if $|s'|=|s|+1$ and $s'[1,|s|]=s$.

By Claim 1, for any $n\in\mathbb{N}$, it follows from $\sigma_2[1,n]\in\Delta_n$ that $\{s_n\in Out^n(q_1,F_{A\!g}):s_n\sim_{A\!g}^\varepsilon \sigma_2[1,n]\}$ is non-empty.
Thus such tree  is infinite.
On the other hand, since $T_2$ is finite branching, it is easy to see that this tree is finite branching.
Therefore, by K$\ddot{o}$nig's lemma, there exists an infinite branch of this tree.
Suppose that this branch is $s_1(=q_1),s_2,s_3,\cdots$.
We define $\sigma_1$ as $\sigma_1[i]\triangleq s_i[end]$ for each $i\in\mathbb{N}$.
By the construction of this tree, we have $\sigma_1\in Out(q_1,F_{A\!g})$ and $\sigma_1\sim^\varepsilon_{A\!g} \sigma_2$.
\qed\end{proof}

Due to symmetry of $\sim^\varepsilon_{A\!g}$, the clauses (1-$b$) and (2-$b$) in the lemma below are redundant.
However, to prove (1-$a$) and (2-$a$) by induction, we need to use induction hypothesis on (1-$b$) and (2-$b$), thus, they are also listed explicitly  in the lemma.
\begin{lemma}\label{lem:characterization 2}
Let $T_i=(S_i,\mathbb{P},\Omega,\Pi_i,\hbar_i)$ be two finite branching alternating transition systems $(i=1,2)$.
Suppose that $d$ is a metric over $\mathbb{P}$, $\varepsilon\in\mathbb{R}^0_+$ and
$Ag\subseteq\Omega$.
For any $q_1\in S_1$, $q_2\in S_2$, $\sigma_1\in(S_1)^{\omega}$ and $\sigma_2\in(S_2)^{\omega}$,

(1) if $q_1\sim^\varepsilon_{A\!g} q_2$, then for any $(\varphi,\gamma)\in H^{\varepsilon}_{A\!g}$,

\ \ \ \ (1-a) $(T_1,d),q_1\models\varphi\Rightarrow (T_2,d), q_2\models \gamma$,

\ \ \ \ (1-b) $(T_2,d),q_2\models\varphi\Rightarrow (T_1,d), q_1\models \gamma$;

(2) if $\sigma_1 \sim^\varepsilon_{A\!g} \sigma_2$, then for any $(\psi,\phi)\in E^{\varepsilon}_{A\!g}$,

\ \ \ \ (2-a) $(T_1,d),\sigma_1 \models\psi\Rightarrow (T_2,d), \sigma_2 \models \phi$,

\ \ \ \ (2-b) $(T_2,d),\sigma_2 \models\psi\Rightarrow (T_1,d), \sigma_1\models \phi$.
\end{lemma}
\begin{proof}
We prove (1) and (2) simultaneously by induction on the ranks of $\varphi$ and $\psi$.

By Definition~\ref{def:complexity of formula}, $\xi_s(\varphi)>0$ and $\xi_p(\psi)>0$ for any state formula $\varphi$ and path formula $\psi$.
So it is clear that if $\xi_s(\varphi)=0$ and $\xi_p(\psi)=0$ then (1) and (2) hold.

Suppose that $(\varphi,\gamma)\in H^{\varepsilon}_{A\!g}$, $(\psi,\phi)\in E^{\varepsilon}_{A\!g}$, $\xi_{s}(\varphi)=\xi_p(\psi)=n+1$ and the conclusions (1) and (2) hold for any $(\varphi_0,\gamma_0)\in H^{\varepsilon}_{A\!g}$ and $(\psi_0,\phi_0)\in E^{\varepsilon}_{A\!g}$ with $\xi_{s}(\varphi_0)=\xi_p(\psi_0)\leq n$.

\textbf{(1-$a$)} By Lemma~\ref{lem:properties of H and E}, $\varphi$ is in one of the following forms: $p$, $\neg\gamma_1, \varphi_1\wedge\varphi_2$, and $\langle\!\langle Ag\rangle\!\rangle\psi_1$.
The argument is split into four cases based on the form of $\varphi$.
In the following, we just consider some sample cases.

 \textbf{Case 1.1} $\varphi=p$ for some $p\in\mathbb{P}$.
Suppose that $q_1\sim^\varepsilon_{A\!g} q_2$ and $(T_1,d),q_1\models\varphi$.
By Lemma~\ref{lem:properties of H and E}, we have $\gamma=\langle\varepsilon\rangle p$.
Then, since $(T_1,d),q_1\models\varphi$ and $q_1\sim^\varepsilon_{A\!g} q_2$, it follows from Definition~\ref{def:semantic of ATLe} and Theorem~\ref{th:approximate bisi} that $\Pi_1(q_1)=p$ and $d(p,\Pi_2(q_2))\leq \varepsilon$.
Further, by  Definition~\ref{def:semantic of ATLe}, we get $(T_2,d), q_2\models \gamma$.

\textbf{Case 1.2} $\varphi=\neg\gamma_0$. Suppose that $q_1\sim^\varepsilon_{A\!g} q_2$ and $(T_1,d),q_1\models\varphi$.
By Lemma~\ref{lem:properties of H and E}, there exists a state formula $\varphi_0$ such that $(\varphi_0,\gamma_0)\in H^{\varepsilon}_{A\!g}$, $\gamma=\neg\varphi_0$ and $\xi_s(\varphi)=\xi_s(\gamma)=\xi_s(\gamma_0)+1$.
It follows from $(T_1,d),q_1\models\varphi$ and Definition~\ref{def:semantic of ATLe} that $(T_1,d),q_1\not\models\gamma_0$.
Further, by induction hypothesis on (1-$b$), we obtain $(T_2,d),q_2\not\models\varphi_0$. Therefore, $(T_2,d),q_2\models\gamma$.

\textbf{Case 1.3} $\varphi=\langle\!\langle Ag\rangle\!\rangle \psi$. Suppose that $q_1\sim^\varepsilon_{A\!g} q_2$ and $(T_1,d),q_1\models\varphi$.
It follows from Lemma~\ref{lem:properties of H and E} that there exists a path formula $\phi$ such that $(\psi,\phi)\in E^{\varepsilon}_{A\!g}$, $\gamma=\langle\!\langle Ag\rangle\!\rangle \phi$ and $\xi_s(\varphi)=\xi_s(\gamma)=\xi_p(\psi)+1$.
Due to $(T_1,d),q_1\models\langle\!\langle Ag\rangle\!\rangle \psi$ and Definition~\ref{def:semantic of ATLe}, there is a strategy $F_{A\!g}$ of $Ag$ such that for any $\sigma\in Out(q_1,F_{A\!g})$, $(T_1,d),\sigma\models\psi$.
It follows from $q_1\sim^\varepsilon_{A\!g} q_2$ and Lemma~\ref{lem:characterization 1} that there exists a strategy $F'_{A\!g}:(S_2)^+\rightarrow 2^{S_2}$  such that
\begin{center}
for any $\sigma'\in Out(q_2,F'_{A\!g})$, $\sigma\sim^\varepsilon_{A\!g}\sigma'$ for some $\sigma\in Out(q_1,F_{A\!g})$.
\end{center}
Then since $(T_1,d),\sigma\models\psi$ for any $\sigma\in Out(q_1,F_{A\!g})$, by induction hypothesis on (2-$a$),
$(T_2,d),\sigma'\models\phi$ for any $\sigma'\in Out(q_2,F'_{A\!g})$.
So we have $(T_2,d),q_2\models\gamma$.

\textbf{(2-$a$)} By Lemma~\ref{lem:properties of H and E} again, $\psi$ is in the form of $\varphi_1$, $\neg\phi_1, \psi_1\wedge\psi_2$, $\mathbf{X}\psi_1$ or $\psi_1\mathbf{U}\psi_2$.
We distinguish five cases based on the form of $\psi$ and just consider some sample cases below.

\textbf{Case 2.1}
$\psi$ is a state formula.
Suppose that $\sigma_1\sim^\varepsilon_{A\!g}\sigma_2$ and $(T_1,d),\sigma_1\models\psi$.
By Lemma~\ref{lem:properties of H and E}, $(\psi,\phi)\in H^{\varepsilon}_{A\!g}$ and $\xi_p(\psi)=\xi_p(\phi)=\xi_s(\psi)+1$.
By Definition~\ref{def:semantic of ATLe}, we have $(T_1,d),\sigma_1[1]\models\psi_1$.
So by induction hypothesis on (1-$a$), $(T_2,d),\sigma_2[1]\models\phi$.
Further, it follows from Definition~\ref{def:semantic of ATLe} that $(T_2,d),\sigma_2\models\phi$.

\textbf{Case 2.2}
$\psi=\psi_1\mathbf{U}\psi_2$.
Suppose that $\sigma_1\sim^\varepsilon_{A\!g}\sigma_2$ and $(T_1,d),\sigma_1\models\psi$.
By Lemma~\ref{lem:properties of H and E}, there exist path formulas $\phi_1$ and $\phi_2$ such that $(\psi_i,\phi_i)\in E^{\varepsilon}_{A\!g} (i=1,2)$, $\phi=\phi_1\mathbf{U}\phi_2$ and $\xi_p(\phi)=\xi_p(\psi)=\max\{\xi_p(\psi_1),\xi_p(\psi_2)\}+1$.
Due to $(T_1,d),\sigma_1\models\psi_1\mathbf{U}\psi_2$ and Definition~\ref{def:semantic of ATLe}, there exists $i\in\mathbb{N}$ such that
\begin{center}
$(T_1,d),\sigma_1[i,\infty]\models\psi_2$ and for any $j<i$, $(T_1,d),\sigma_1[j,\infty]\models\psi_1$.
\end{center}
By induction hypothesis on (2-$a$), it follows that $(T_2,d),\sigma_2[i,\infty]\models\phi_2$ and for any $j<i$, $(T_2,d),\sigma_2[j,\infty]\models\phi_1$.
Then by Definition~\ref{def:semantic of ATLe}, we obtain $(T_2,d),\sigma_2\models\phi$.
\qed\end{proof}

\begin{lemma}\label{lem:characterization 3}
Let $T_i=(S_i,\mathbb{P},\Omega,\Pi_i,\hbar_i)$ be two finite branching alternating transition systems $(i=1,2)$.
Suppose that $d$ is a metric over $\mathbb{P}$, $\varepsilon\in\mathbb{R}^0_+$ and
$Ag\subseteq\Omega$.
For any $q_1\in S_1$ and $q_2\in S_2$,  $q_1\sim^\varepsilon_{A\!g} q_2$ if they satisfy the following two conditions:

(1) for any $(\varphi,\gamma)\in H^\varepsilon_{A\!g}$, $(T_1,d),q_1\models\varphi\Rightarrow(T_2,d),q_2\models\gamma$,

(2) for any $(\varphi,\gamma)\in H^\varepsilon_{A\!g}$, $(T_2,d),q_2\models\varphi\Rightarrow(T_1,d),q_1\models\gamma$.
\end{lemma}
\begin{proof}
Set
\begin{equation*}
R=\{(q_1,q_2):q_1\in S_1 \textrm{ and }q_2\in S_2 \textrm{ satisfy the above conditions (1) and (2)}\}.
\end{equation*}
To complete the proof, it is enough to show that $R$ is an $(Ag,\varepsilon)$-alternating approximate bisimulation.
Suppose that $R$ is not.
Since each pair in $R$ satisfies the condition (1) in Definition~\ref{def:approximate bisimi}, there exists $(q_1,q_2)\in R$ satisfying one of the following conditions:

(i) $\exists Q_1\in\hbar_1(q_1,Ag)\forall Q_2\in\hbar_2(q_2,Ag) \exists Q'_2\in \hbar_2(q_2,\Omega-Ag) \forall Q'_1\in\hbar_1(q_1,\Omega-Ag)((Q_1\cap Q'_1)\times (Q_2\cap Q'_2)\not\subseteq R)$,

(ii) $\exists Q_2\in\hbar_2(q_2,Ag)\forall Q_1\in\hbar_1(q_1,Ag) \exists Q'_1\in \hbar_1(q_1,\Omega-Ag) \forall Q'_2\in\hbar_2(q_2,\Omega-Ag)((Q_1\cap Q'_1)\times (Q_2\cap Q'_2)\not\subseteq R)$.

W.l.o.g, suppose that (i) holds.
It follows from  (i) and Definition~\ref{def:ATS} that for any $Q_2\in\hbar_2(q_2,Ag)$, there exists $Q'_2\in\hbar_2(q_2,\Omega-Ag)$ and  $q'_2\in Q_2$ such that $\{q'_2\}=Q_2\cap Q'_2$ and for any $Q'_1\in\hbar(q_1,\Omega-Ag)$,
\begin{center}
$\{q'_1\}=Q_1\cap Q'_1$ and $(q'_1,q'_2)\not\in R$ for unique $q'_1$.
\end{center}

Hence, for each $Q_2\in\hbar_2(q_2,Ag)$ and $Q'_1\in\hbar_1(q_1,\Omega-Ag)$, we can fix a pair of states $q_{Q_2}\in Q_2$ and $q_{Q'_1}\in Q_1$ satisfying the above conditions.
Then for any $Q_2\in\hbar_2(q_2,Ag)$ and $Q'_1\in\hbar_1(q_1,\Omega-Ag)$, it follows from $(q_{Q'_1},q_{Q_2})\not\in R$ that there exists $(\varphi,\gamma)\in H^\varepsilon_{A\!g}$ such that
\begin{center}
(a) $(T_1,d),q_{Q'_1}\models\varphi$ but $(T_2,d),q_{Q_2}\models\neg\gamma$, or

(b)$(T_2,d),q_{Q_2}\models\varphi$ but $(T_1,d),q_{Q'_1}\models\neg\gamma$.
\end{center}

If (a) holds, then $(T_2,d),q_{Q_2}\models\neg\gamma$ but $(T_1,d),q_{Q'_1}\models\neg\neg\varphi$.
Moreover, by Definition~\ref{def:relation of formula 1}, we get $(\neg\gamma,\neg\varphi)\in H^\varepsilon_{A\!g}$.
Thus for any $Q_2\in\hbar_2(q_2,Ag)$ and $Q'_1\in\hbar_1(q_1,\Omega-Ag)$, there exists $(\varphi_{Q'_1,Q_2},\gamma_{Q'_1,Q_2})\in H^\varepsilon_{A\!g}$ such that
\begin{equation}\label{eq:char 3.1}
(T_2,d),q_{Q_2}\models\varphi_{Q'_1,Q_2} \textrm{ but } (T_1,d),q_{Q'_1}\models\neg\gamma_{Q'_1,Q_2}.
\end{equation}

So, for any $Q^*_2\in\hbar_2(q_2,Ag)$,  it follows that
\begin{center}
$(T_2,d),q_{Q^*_2}\models\bigvee_{Q_2\in\hbar_2(q_2,Ag)}\bigwedge_{Q'_1\in\hbar_1(q_1,\Omega-Ag)}\varphi_{Q'_1,Q_2}$\footnote{$\bigvee_{i\in I}\varphi_i$ and $\bigwedge_{i\in I}\varphi_i$ can be defined as usual, where $I$ is a finite index set.}.
\end{center}
By Lemma~\ref{lem:stra_Ag}, for any strategy $F'_{A\!g}:(S_2)^+\rightarrow 2^{S_2}$, we have  $F'_{A\!g}(q_2)\in\hbar_2(q_2,Ag)$ and then it follows from Definition~\ref{def:outcome} that for some $\sigma_2\in Out(q_2,F'_{A\!g})$,
\begin{center}
$\sigma_2[2]= q_{F'_{A\!g}(q_2)}$ and $(T_2,d),\sigma_2[2]\models\bigvee_{Q_2\in\hbar_2(q_2,Ag)}\bigwedge_{Q'_1\in\hbar_1(q_1,\Omega-Ag)}\varphi_{Q'_1,Q_2}$.
\end{center}
Further, by Definition~\ref{def:semantic of ATLe}, we obtain
\begin{center}
$(T_2,d),q_2\models\neg\langle\!\langle Ag\rangle\!\rangle\neg\mathbf{X}\bigvee_{Q_2\in\hbar_2(q_2,Ag)}\bigwedge_{Q'_1\in\hbar_1(q_1,\Omega-Ag)}\varphi_{Q'_1,Q_2}$.
\end{center}

Moreover, it follows from Definition~\ref{def:relation of formula 1} that $(\varphi^*,\gamma^*)\in H_{A\!g}^{\varepsilon}$, where
\begin{center}
$\varphi^*=\neg\langle\!\langle Ag\rangle\!\rangle\neg\mathbf{X}\bigvee_{Q_2\in\hbar_2(q_2,Ag)}\bigwedge_{Q'_1\in\hbar_1(q_1,\Omega-Ag)}\varphi_{Q'_1,Q_2}$
\end{center}
 and
\begin{center}
 $\gamma^*=\neg\langle\!\langle Ag\rangle\!\rangle\neg\mathbf{X}\bigvee_{Q_2\in\hbar_2(q_2,Ag)}\bigwedge_{Q'_1\in\hbar_1(q_1,\Omega-Ag)}\gamma_{Q'_1,Q_2}$.
\end{center}

Hence, due to $(T_2,d),q_2\models\varphi^*$ and $(q_1,q_2)\in R$, we get
\begin{equation}\label{eq:char 3.2}
(T_1,d),q_1\models \gamma^*.
\end{equation}

On the other hand, by Definition~\ref{def:strategy}, it is clear that $F_{A\!g}(q_1)=Q_1\in\hbar_1(q_1,Ag)$ for some strategy $F_{A\!g}:(S_1)^+\rightarrow 2^{S_1}$.
Then by~(\ref{eq:char 3.2}), there exists $\sigma_1\in Out(q_1,F_{A\!g})$ such that
\begin{center}
$(T_1,d),\sigma_1[2]\models \bigvee_{Q_2\in\hbar_2(q_2,Ag)}\bigwedge_{Q'_1\in\hbar_1(q_1,\Omega-Ag)}\gamma_{Q'_1,Q_2}$.
\end{center}

By Definition~\ref{def:outcome}, $\{\sigma_1[2]\}=Q_1\cap Q''_1$ for some $Q''_1\in\hbar_1(q_1,\Omega-Ag)$.
Clearly, $\sigma_1[2]=q_{Q''_1}$.
Then  $(T_1,d),q_{Q''_1}\models \bigwedge_{Q'_1\in\hbar_1(q_1,\Omega-Ag)}\gamma_{Q'_1,Q^*_2}$ for some $Q^*_2\in\hbar_2(q_2,Ag)$.
Thus we have $(T_1,d),q_{Q''_1}\models \gamma_{Q''_1,Q^*_2}$, which contradicts (\ref{eq:char 3.1}).
\qed\end{proof}

Now, we arrive at the main result of this section, which offers a logical characterization of alternating approximate bisimilariy.

\begin{theorem}\label{th:main}
Let $T_i=(S_i,\mathbb{P},\Omega,\Pi_i,\hbar_i)$ be two finite branching alternating transition systems $(i=1,2)$.
Suppose that $d$ is a metric over $\mathbb{P}$, $\varepsilon\in\mathbb{R}^0_+$ and
$Ag\subseteq\Omega$.
For any $q_1\in S_1$ and $q_2\in S_2$, $q_1\sim^\varepsilon_{A\!g} q_2$ if and only if they satisfy the following conditions:

(1) for any $(\varphi,\gamma)\in H^\varepsilon_{A\!g}$, $(T_1,d),q_1\models\varphi\Rightarrow(T_2,d),q_2\models\gamma$,

(2) for any $(\varphi,\gamma)\in H^\varepsilon_{A\!g}$, $(T_2,d),q_2\models\varphi\Rightarrow(T_1,d),q_1\models\gamma$.
\end{theorem}
\begin{proof}
Immediately follows from Lemma~\ref{lem:characterization 2} and~\ref{lem:characterization 3}.
\qed\end{proof}

For $\varepsilon=0$, by the clause (1) in Lemma~\ref{lem:prop of Ag-E}, $(Ag,\varepsilon)$-alternating approximate bisimilarity is an equivalence relation.
In this case, the above result degenerates into one in the usual style.

\begin{corollary}\label{cor:characterization of 0}
Let $T_i=(S_i,\mathbb{P},\Omega,\Pi_i,\hbar_i)$ be two finite branching alternating transition systems $(i=1,2)$.
Suppose that $d$ is a metric over $\mathbb{P}$ and
$Ag\subseteq\Omega$. For any $q_1\in S_1$ and $q_2\in S_2$, $q_1\sim_{A\!g}^0q_2$ if and only if for any $Ag$-ATL$_0$ state formula $\varphi$,
\begin{center}
$(T_1,d),q_1\models\varphi\Leftrightarrow (T_2,d),q_2\models\varphi.$
\end{center}
\end{corollary}
\begin{proof} Suppose that $q_1\in S_1$ and $q_2\in S_2$.
Then we have

\ \ \ \ $q_1\sim^0_{A\!g}q_2$

iff for any $(\varphi,\gamma)\in H^0_{A\!g}$, $(T_1,d),q_1\models\varphi\Rightarrow (T_2,d),q_2\models\gamma$, and vice verse

\ \ \ \ \ \ \ \ \ \ \ \ \ \ \ (by Theorem~\ref{th:main})

iff for any $(\varphi,\gamma)\in H^0_{A\!g}$, $(T_1,d),q_1\models\varphi\Rightarrow (T_2,d),q_2\models\varphi$,  and vice verse

\ \ \ \ \ \ \ \ \ \ \ \ \ \ \ (by (1) in Lemma~\ref{lem:basic properties of H and E})

iff for any $(\varphi,\gamma)\in H^0_{A\!g}$, $(T_1,d),q_1\models\varphi\Leftrightarrow (T_2,d),q_2\models\varphi$

iff for any $Ag$-ATL$_0$ state formula $\varphi$, $(T_1,d),q_1\models\varphi\Leftrightarrow (T_2,d),q_2\models\varphi$.

\ \ \ \ \ \ \ \ \ \ \ \ \ \ \ (by (3) in Lemma~\ref{lem:basic properties of H and E})
\qed\end{proof}

As mentioned in Section~\ref{sec:pre}, a logical characterization of alternating bisimilarity in terms of ATL has been provided in~\cite{AlurCONCUr98}.
In the following, we show that this characterization can be obtained from the above result immediately.
The syntax and semantic  of ATL is similar to ATL$_\varepsilon$ and the only difference between them is that ATL does not refer to the modal operator $\langle\varepsilon\rangle$ and metric over observations.
Here we do not recall ATL formally, which  can be found in~\cite{alur2002alternating}\cite{AlurCONCUr98}.
Since the semantics of ATL has nothing to do with metric over observations, we use $T,q\models\varphi$ to denote that the state $q$ of $T$ satisfies ATL formula $\varphi$.

\begin{corollary}\label{cor:characterization}
Let $T_i=(S_i,\mathbb{P},\Omega,\Pi_i,\hbar_i)$ be two finite branching alternating transition systems $(i=1,2)$ and
$Ag\subseteq\Omega$. For any $q_1\in S_1$ and $q_2\in S_2$, $q_1\sim_{A\!g}q_2$ if and only if for any $Ag$-ATL state formula $\varphi$,
\begin{center}
$T_1,q_1\models\varphi\Leftrightarrow T_2,q_2\models\varphi.$
\end{center}
\end{corollary}
\begin{proof} It is not difficult to see that two states are logical equivalent w.r.t $Ag$-ATL$_0$ state formulas if and only if  so are they w.r.t $Ag$-ATL state formulas.
Thus by the clause (1) in Lemma~\ref{lem:prop of Ag-E} and Corollary~\ref{cor:characterization of 0}, the conclusion holds.
\qed\end{proof}

\section{Application of modal characterization: temporal logical control}
\label{sec:finite abstraction}
For control systems with disturbances, Pola and Tabuada adopt infinite alternating transition systems to model their sampling systems and construct finite alternating transition systems as their finite abstractions~\cite{pola2008symbolic}\cite{pola2009symbolic}.
In these work, alternating approximate bisimilarity is introduced to capture the equivalence between these sampling systems and finite abstractions.
Based on results obtained in the previous section, we will establish a relationship between linear temporal logical specifications which are satisfied by these sampling systems under control and by the corresponding finite abstractions under control, respectively.
Moreover, we give a potential application of this result in the linear temporal logical control of control systems with disturbances.

\subsection{Control systems and its finite abstractions}
This subsection recalls some notions and results about control systems with disturbances and their finite abstractions provided by~\cite{pola2008symbolic}\cite{pola2009symbolic}.
Before doing so, we introduce some useful notations.

Given a vector $\textit{x}\in\mathbb{R}^{n}$,
we denote by $\textit{x}_{i}$ the $\textit{i}$-th element of $\textit{x}$ and $\|x\|\triangleq \max \{|\textit{x}_{1}|, |\textit{x}_{2}|, \cdots, |\textit{x}_{n}|\}$ where $|\textit{x}_{i}|$ is the absolute value of $\textit{x}_{i}$.
The set $X\subseteq \mathbb{R}^{n}$ is said to be bounded if and only if $ \sup\{\|x\|:x\in X\}<\infty$.
For any measurable function $f:\mathbb{R}_{+}^{0}\rightarrow \mathbb{R}$, $\|f\|_{\infty}\triangleq\sup\{\|f(t)\|,t\geq 0\}$ and $f$ is said to be essentially bounded if $\|f\|_{\infty}<\infty$.
For a given time $\tau\in\mathbb{R}_{+}$, define $f_{\tau}$ so that $f_{\tau}(t)=f(t)$ for any $t\in[0,\tau)$, and $f(t)=0$ elsewhere; $f$ is said to be locally essentially bounded if for any $\tau\in\mathbb{R}_{+}$, $f_{\tau}$ is essentially bounded.
In this section, we consider the metric $d$ on $\mathbb{R}^{n}$ defined as ${d}(x,y)=\max\{|x_1-y_1|,|x_2-y_2|,\cdots |x_n-y_n|\}$.

\begin{definition}\label{def:control system}\cite{pola2008symbolic}\cite{pola2009symbolic}
A control system with disturbances is a quadruple $\Sigma=(X,W,\mathcal{W},f)$, where

$\bullet$ $X\subseteq\mathbb{R}^n$ is the state space;

$\bullet$ $W=U\times V$ is the input space, where

\ \ \ $U\subseteq \mathbb{R}^m$ is the control input space;

\ \ \ $V\subseteq \mathbb{R}^s$ is the disturbance input space;

$\bullet$ $\mathcal{W}$ is a subset of the set of all measurable and locally essentially bounded functions of time from intervals of the form $]a,b[\subseteq \mathbb{R}$ to $W$ with $a<0$ and $b>0$;

$\bullet$ $f:X\times W\rightarrow X$ is a continuous map satisfying the following Lipschitz assumption: for every compact set $K\subset X$, there exists a constant $\kappa>0$ such that
\begin{equation*}
||f(x,w)-f(y,w)||\leq \kappa||x-y||
\end{equation*}
for all $x,y\in K$ and all $w\in W$.

A locally absolutely continuous curve $\mathbf{x}:]a,b[\rightarrow X$
is said to be a trajectory if there exists $\mathbf{w}\in\mathcal{W}$ satisfying $\mathbf{\dot{x}}(t)=f(\mathbf{x}(t),\mathbf{w}(t))$ for almost all $t\in ]a,b[$.

A control system is said to be forward complete if and only if every trajectory is defined on an interval of the form $]a,\infty[$.
\end{definition}

\textbf{Convention.} \textit{As in \cite{pola2008symbolic}\cite{pola2009symbolic}, we assume that  $0\in X$, $X\subseteq \mathbb{R}^{n}$ is a bounded polytopic sets with non-empty interior, and the control system $\Sigma$ is forward complete }.

For such systems, Pola and Tabuada adopt a variety of alternating transition systems as models of their sampling systems and finite abstractions~\cite{pola2008symbolic}\cite{pola2009symbolic}.

\begin{definition}\label{def:var of ATS}
An alternating transition system is a tuple
$T=(S,A,B,\longrightarrow,\mathbb{P},\Pi)$
consisting of a set of states $S$, a set of control labels $A$, a set of disturbance labels $B$, a transition relation $\rightarrow\subseteq S\times A\times B\times S$, an observation set $\mathbb{P}$ and an observation function $\Pi:S\rightarrow \mathbb{P}$.

We say that an alternating transition system $T$ is metric if the observation set $\mathbb{P}$ is equipped with a metric, $T$ is non-blocking if $\{q':q\xrightarrow{a,b}q'\}\not=\emptyset$ for any $q\in Q$, $a\in A$ and $b\in B$, and
$T$ is finite if $S$, $A$ and $B$ are finite.
An infinite sequence $\sigma\in S^{\omega}$ is said to be a trajectory of $T$ if and only if for all $i\in\mathbb{N}$, $\sigma[i]\xrightarrow{a_{i},b_{i}} \sigma[i+1]$ for some $a_{i}\in A $ and $b_{i}\in B$.
\end{definition}

We may view the above alternating transition system as a variant of one defined in Definition~\ref{def:ATS}.
The differences between them lie in: the above notion involves only two agents which choose successor states by means of choosing inputs, moreover, successor states of a given state may not be determined even if these two agents make choices.
In this section, following Pola and Tabuada, the notion ``alternating transition system" refers to the one defined above.
Similar to Definition~\ref{def:strategy} and~\ref{def:outcome}, the strategy and the corresponding outcomes of these systems are defined below.

\begin{definition}\label{def:strategy of var}
A control strategy for an alternating transition system $T=(S,A,B,\longrightarrow,\mathbb{P},\Pi)$ is a function $F:S^{+}\rightarrow 2^{A}-\{\emptyset\}$. For any $q\in S$, the outcomes $Out^{n}_{T}(q,F)$ $(n\in\mathbb{N})$ and $Out_{T}(q,F)$  of $F$ from $q$ are defined as follows:
\begin{equation*}
\begin{aligned}
Out^{n}_{T}(q,F) =\{s\in S^{n}:  & s[1]=q\ \mathrm{and}
\\ &\forall 1\leq i<n \exists a_{i}\in F(s[1,i]) \exists b_{i}\in B (s[i]\xrightarrow{a_{i},b_{i}} s[i+1])\},
\end{aligned}
\end{equation*}
\begin{equation*}
\begin{aligned}
 Out_{T}(q,F)=\{\sigma\in S^{\omega}: & \sigma[1] =q\ \mathrm{and}
  \\ &\forall i\in\mathbb{N} \exists a_{i}\in F(\sigma [1,i]) \exists b_{i}\in B (\sigma [i]\xrightarrow{a_{i},b_{i}} \sigma[i+1])\}.
  \end{aligned}
\end{equation*}
\end{definition}

The notion of alternating approximate bisimilarity provided by Pola and Tabuada is recalled below.
It is not difficult to see that such notion and one in Definition~\ref{def:approximate bisimi} are the same in spirit.

\begin{definition}\label{def:var of AAB}\cite{pola2009symbolic}
Let $T_{i}=(S_{i},A_{i}, B_{i}, \xrightarrow{}_i, \mathbb{P}, \Pi_{i})$ ($i=1,2$) be two metric, non-blocking alternating transition systems and let $d$ be a metric over $\mathbb{P}$.
Given a precision $\varepsilon \in \mathbb{R}_{+}$, a relation $R\subseteq S_{1} \times S_{2}$ is said to be an alternating $\varepsilon$-approximate ($A\varepsilon A$) bisimulation relation between $T_{1}$ and $T_{2}$ if for any $(q_{1},q_{2})\in R$,

(i) $d(\Pi_{1}(q_{1}),\Pi_{2}(q_{2}))\leq\varepsilon$;

(ii) $\forall a_{1}\in A_{1} \exists a_{2}\in A_{2} \forall b_{2} \in B_{2} \forall q'_{2}\in S_{2} (q_{2}\xrightarrow{a_{2},b_{2}}_2 q'_{2} \Rightarrow \exists b_{1}\in B_{1} {\exists q'_{1}\in S_{1}}\\ ({q_{1}\xrightarrow{a_{1},b_{1}}_1 q'_{1}} \textrm{ and }  (q'_{1},q'_{2})\in R))$.

(iii) $\forall a_{2}\in A_{2} \exists a_{1}\in A_{1} \forall b_{1} \in B_{1} \forall q'_{1}\in S_{1} (q_{1}\xrightarrow{a_{1},b_{1}}_1 q'_{1} \Rightarrow \exists b_{2}\in B_{2} \exists q'_{2}\in S_{2}\\  ({q_{2}\xrightarrow{a_{2},b_{2}}_2 q'_{2}} \textrm{ and }  (q'_{1},q'_{2})\in R))
$.

For any $q_{1}\in S_{1}$ and $q_{2}\in S_{2}$, they are said to be $A\varepsilon A$ bisimilar, in symbols $q_{1}\sim_{\varepsilon} q_{2}$,  if there exists an $A\varepsilon A$ bisimulation relation $R$ between $T_{1}$ and $T_{2}$ such that $(q_{1},q_{2})\in R$.
Moreover, $T_{1}$ and $T_{2}$ are said to be $A\varepsilon A$ bisimilar, in symbols $T_{1}\simeq_{\varepsilon} T_{2}$,  if there exists an $A\varepsilon A$ bisimulation relation $R$ between $T_{1}$ and $T_{2}$ such that $S_{1}=\{q_{1}\in S_{1}: (q_{1},q_{2})\in R\ \mathrm{for}\ \mathrm{some}\ q_{2}\in S_{2}\}$ and $S_{2}=\{q_{2}\in S_{2}: (q_{1},q_{2})\in R\ \mathrm{for}\ \mathrm{some}\ q_{1}\in S_{1}\}$.
\end{definition}

For control systems with disturbances, Pola and Tabuada construct infinite and finite non-blocking alternating transition systems as their samples and finite abstractions, respectively.
The detailed construction is referred to~\cite{pola2009symbolic}.
Moreover, they demonstrate that under some assumption, the sample $T_\tau(\Sigma)$ and finite abstraction are alternating approximate bisimilar.

\begin{theorem}~\cite{pola2009symbolic}\label{Th:bisi}
Given a control system $\Sigma=(X,U\times V,\mathcal{W},f)$, if $\Sigma$ is $\delta$-GAS and $U\times V$ is compact, then for any desired precision $\varepsilon\in\mathbb{R}_+$, there exist $\tau\in\mathbb{R}_+$ and a finite abstraction $T$ of $\Sigma$ that is $A\varepsilon A$ bisimilar to the sampling system $T_{\tau}(\Sigma)$ of $\Sigma$\footnote{The definition of $\delta$-GAS can be found in~\cite{pola2009symbolic}.}.
\end{theorem}

For convenience, we set $T_{\varepsilon,\tau}(\Sigma)\triangleq\{T: T$  is a finite abstraction of $\Sigma$ that is $A\varepsilon A$ bisimilar to the sampling system $T_{\tau}(\Sigma)\}$.
\subsection{Logical specifications satisfied by samples and abstractions}
In recent years, temporal logic, due to its resemblance to natural language and the existence of algorithms for model checking, is widely adopted to describe the desired specifications of control systems. For example, linear temporal logic (LTL) is used to express specifications of discrete-time linear systems [8] and continuous-time linear systems [7].
On the other hand, as mentioned in Introduction, finite abstractions of control systems often are adopted to the analysis and design of control systems.
Then a natural question arises at this point: what is the relationship between linear temporal logical specifications which are satisfied under control by sampling systems and by the corresponding finite abstractions respectively?
This subsection intends to consider such question.
To this end, we introduce linear temporal logic LTL$^{\varepsilon}_+$ as follows.

\begin{definition}
Let $\mathbb{P}$ be a finite set of propositions and $\varepsilon\in \mathbb{R}_+$.
LTL$^{\varepsilon}_+(\mathbb{P})$ formulas  are defined inductively as:
\begin{equation*}
\phi::=p|\langle\varepsilon\rangle p|\phi_1\vee\phi_2|\phi_1\wedge\phi_2|\mathbf{X}\phi|\phi_1\mathbf{U}\phi_2, \textrm{ where }p\in \mathbb{P}.
\end{equation*}

For any LTL$^{\varepsilon}_+(\mathbb{P})$ formula $\phi$, if $\langle\varepsilon\rangle$ does not occur in $\phi$, then $\phi$ is said to be a LTL$_+(\mathbb{P})$ formula.
\end{definition}

As usual, if $\mathbb{P}$ is clear from the context, we abbreviate LTL$^{\varepsilon}_+(\mathbb{P})$ (LTL$_+(\mathbb{P})$) to LTL$^{\varepsilon}_+$ (LTL$_+$, respectively).

\begin{definition}\label{def:semantic for LTL}
Let $T=(S,A, B, \xrightarrow{}, \mathbb{P}, \Pi)$ be an alternating transition system,  $\varepsilon\in \mathbb{R}_+$ and let $d$ be a metric over $\mathbb{P}$.
The satisfaction relation $\models\subseteq S^{\omega}\times LTL^{\varepsilon}_+(\mathbb{P})$ is inductively defined as:

$\bullet$ $(T,d),\sigma\models p$ iff $\Pi(\sigma[1])=p$;

$\bullet$ $(T,d),\sigma\models\langle\varepsilon\rangle p$ iff $d(p,\Pi(\sigma[1]))\leq\varepsilon$;

$\bullet$ $(T,d),\sigma\models\phi_1\vee\phi_2$ iff $(T,d),\sigma\models\phi_1$ or $(T,d),\sigma\models\phi_2$;

$\bullet$ $(T,d),\sigma\models\phi_1\wedge\phi_2$ iff $(T,d),\sigma\models\phi_1$ and $(T,d),\sigma\models\phi_2$;

$\bullet$ $(T,d),\sigma\models\mathbf{X}\phi$ iff $(T,d),\sigma[2,\infty]\models\phi$;

$\bullet$ $(T,d),\sigma\models\phi_1\mathbf{U}\phi_2$ iff there exists $j\in\mathbb{N}$ such that $(T,d),\sigma[j,\infty]\models\phi_2$ and for any $i<j$, $(T,d),\sigma[i,\infty]\models\phi_1$.
\end{definition}

Obviously, LTL$^{\varepsilon}_+$ can be viewed as a sublanguage of ATL$_{\varepsilon}$.
In particular, by the above definitions and Definition~\ref{def:logic app} and~\ref{def:semantic of ATLe}, each LTL$^{\varepsilon}_+$ formula can be seen as a path formula of ATL$_{\varepsilon}$.

Inspired by Definition~\ref{def:relation of formula 1}, we introduce a transformation function below.

\begin{definition}\label{def:transform}
Let $\mathbb{P}$ be a set of propositions and $\varepsilon\in \mathbb{R}_+$.
The transformation $Tr_\varepsilon$ mapping LTL$_+$ formulas to LTL$_+^\varepsilon$ formulas is inductively defined as follows:

$\bullet$ $Tr_\varepsilon(p)=\langle\varepsilon\rangle p$ for any $p\in\mathbb{P}$;

$\bullet$ $Tr_\varepsilon(\phi_1\vee\phi_2)=Tr_\varepsilon(\phi_1)\vee Tr_\varepsilon(\phi_2)$;

$\bullet$ $Tr_\varepsilon(\phi_1\wedge\phi_2)=Tr_\varepsilon(\phi_1)\wedge Tr_\varepsilon(\phi_2)$;

$\bullet$ $Tr_\varepsilon(\mathbf{X}\phi)=\mathbf{X}Tr_\varepsilon(\phi)$;

$\bullet$ $Tr_\varepsilon(\phi_1\mathbf{U}\phi_2)=Tr_\varepsilon(\phi_1)\mathbf{U} Tr_\varepsilon(\phi_2)$.
\end{definition}

Clearly, the graph of such transformation is a subrelation of $E^\varepsilon_{A\!g}$ (see Definition~\ref{def:relation of formula 1}), that is

\begin{proposition}\label{pro:def of LTL-E}
For any LTL$_+$ formula $\phi$, $(\phi,Tr_\varepsilon(\phi))\in E^\varepsilon_{A\!g}$.
\end{proposition}
\begin{proof}
Follows from Definition~\ref{def:relation of formula 1} and~\ref{def:transform}.\qed
\end{proof}

Then, by Proposition~\ref{pro:basic properties of H and E}, for each pair of formulas $\phi$ and $Tr_\varepsilon(\phi)$, if a state sequence satisfies $\phi$ then it satisfies $Tr_\varepsilon(\phi)$.
But the converse of this result fails in general.
So, given $\phi$ and $Tr_\varepsilon(\phi)$, when considering them as  specifications, we may view $Tr_\varepsilon(\phi)$ as a \textit{looser} version of specification $\phi$.

Similar to Lemma \ref{lem:characterization 2}, we may prove the following result.
It should be pointed out that although there exist some differences between notions involved in Lemma~\ref{lem:characterization 2} and Proposition~\ref{pro:control systems' property} (see Definition~\ref{def:ATS} and~\ref{def:var of ATS}, Definition~\ref{def:approximate bisimi} and~\ref{def:var of AAB}), the latter may be proved analogously to the former.

\begin{proposition}\label{pro:control systems' property}
Let $T_{1}=(S_{1},A_{1}, B_{1}, \xrightarrow{}_1, \mathbb{P}, \Pi_{1})$ be an infinite, non-blocking alternating transition system and $T_{2}=(S_{2},A_{2}, B_{2}, \xrightarrow{}_2, \mathbb{P}, \Pi_{2})$ be a finite, non-blocking alternating transition system.
Suppose that $d$ is a metric over $\mathbb{P}$ and $\varepsilon\in\mathbb{R}_+$.
For any $q_1\in S_1$ and $q_2\in S_2$ with $q_1\sim_\varepsilon q_2$ and for any LTL$_+$ formula $\phi$, if there exists a strategy $F:S^+\rightarrow 2^{A_2}$ of $T_2$ such that $(T_2,d),\sigma\models\phi$ for any $\sigma\in Out(q_2,F)$, then so does $T_1$ for $Tr_\varepsilon(\phi)$.
\end{proposition}

Immediately, we have
\begin{corollary}\label{cor:under control}
Let  $\Sigma=(X,U\times V,\mathcal{W},f)$ be a $\delta$-GAS control system with compact input space $U\times V$ and $\varepsilon\in\mathbb{R}_+$.
Then there exists $\tau\in\mathbb{R}_+$ and a finite abstraction $T_f\in T_{\varepsilon,\tau}(\Sigma)$ that is A$\varepsilon$A bisimilar to the sampling system $T_{\tau}(\Sigma)$, and for any $T\in T_{\varepsilon,\tau}(\Sigma)$ and LTL$_+$ formula $\phi$, if there exists a state $q$ and a strategy $F$ of $T$ such that $(T,d),\sigma\models \phi$ for any $\sigma\in Out(q,F)$, then so does $T_{\tau}(\Sigma)$ for $Tr_\varepsilon(\phi)$.
\end{corollary}
\begin{proof}
Follows from Proposition~\ref{pro:control systems' property} and Theorem~\ref{Th:bisi}.
\qed
\end{proof}

Due to the above result, for a control system with disturbances satisfying the conditions mentioned in Corollary~\ref{cor:under control}, given a LTL$_+$ formula $\phi$ as a specification, if its finite abstraction and its sampling system are $A\varepsilon A$ bisimilar and the former satisfies specification $\phi$ under control, then so does the latter for a looser specification $Tr_\varepsilon(\phi)$.
\subsection{Linear temporal logical control}
Recently, finite abstraction and the notion of bisimilarity have attracted some people's notice in the area of analysis and design of control systems~\cite{alur2000discrete}\cite{fainekos2005hybrid}\cite{tabuada2003discrete}\cite{tabuada2006linear}.
In general, control systems and their finite abstractions share properties of interest if they are bisimilar.
In particular, according to modal characterization of bisimilarity, they satisfy same temporal logical properties.
Moreover, the analysis and design of finite abstractions is simpler than that of control systems.
Thus the analysis and design of control systems can be equivalently performed on the corresponding finite abstractions.
\begin{figure}[t]
\begin{center}
\centerline{\includegraphics[scale=0.6]{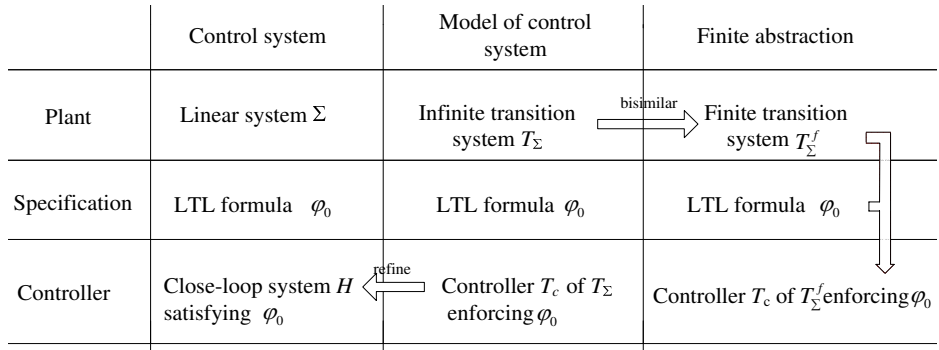}}\vspace{-0.6in}
\end{center}
\hfil\break
\setlength{\abovecaptionskip}{10pt}
\setlength{\belowcaptionskip}{10pt}
\caption{Controller design \cite{tabuada2006linear}: non-disturbance case}\label{fig:illust0}
\end{figure}

As an example, Fig~\ref{fig:illust0} illustrates the function of finite abstraction and bisimilarity in the formal design of linear discrete system \cite{tabuada2006linear}.
Given a linear discrete system $\Sigma$, Tabuada and Pappas provide an infinite transition system $T_{\Sigma}$ as the formal model of $\Sigma$ and construct a finite transition system $T_{\Sigma}^{f}$ as the finite abstraction of $\Sigma$.
They prove the following result which lays the foundation of the design method of controllers presented in \cite{tabuada2006linear}.
\begin{center}
\textit{$T_{\Sigma}$ and $T_{\Sigma}^{f}$ are bisimilar and then share the same properties described by linear temporal logic.}~($*$)
\end{center}
Thus, given an LTL specification $\varphi_0$, the formal design of $T_{\Sigma}$ can be equivalently performed on the finite abstraction $T_{\Sigma}^{f}$.
Tabuada and Pappas construct a controller $T_c$ of $T_{\Sigma}^{f}$ enforcing $\varphi_0$ and demonstrate that $T_{\Sigma}$ satisfies $\varphi_0$ under this controller as well.
Furthermore, based on this controller, a close-loop system $H$ satisfying $\varphi_0$ is generated.
Similar methods are also adopted in ~\cite{fainekos2005hybrid}\cite{kloetzer2008fully}\cite{tabuada2003discrete}.

It is worth to be pointed out that the work~\cite{alur2000discrete}\cite{fainekos2005hybrid}\cite{tabuada2003discrete}\cite{tabuada2006linear} consider only the non-disturbance control systems.
We intend to generalize these methods to the disturbance case.
Similar to the conclusion ($*$) above, as illustrated by Fig~\ref{fig:graph2}\footnote{In this figure, ATS is the abbreviation of alternating transition system.}, Corollary~\ref{cor:under control} in this paper combining with the work in~\cite{pola2009symbolic}  provides analogous  results for linear temporal logical control of control systems with disturbances.
In detail, Pola and Tabuada construct finite abstractions of control systems that are $A\varepsilon A$ bisimilar to the samples of control systems~\cite{pola2009symbolic}, while we demonstrate that if finite abstraction satisfies a specification $\phi$ under control then so does the samples for a looser specification $Tr_\varepsilon(\phi)$ (see Corollary~\ref{cor:under control}).
These results inspire us to provide an approach for the design of control system as shown in Fig~\ref{fig:graph2}:
first, construct finite abstraction that is $A\varepsilon A$ bisimilar to the sample of control system;
second, find a strategy of finite abstraction enforcing the given LTL$_+$ specification $\phi$; and finally, construct controller for control systems based on this strategy so that sampling system satisfies the transformed specification $Tr_\varepsilon(\varphi)$ \footnote{Since we often just can  observe the sampling systems rather than control systems with disturbances, it may be reasonable to require that the samples satisfy specifications under such controllers.}.
The first step has been completed by Pola and Tabuada~\cite{pola2009symbolic}.
For the second step, an algorithm has been offered to find strategies of alternating transition systems enforcing linear temporal logical specifications~\cite{kloetzer2008dealing} and this algorithm can be adopted to obtain the desired strategies for finite abstractions.
So there is only one question left to answer: how to construct the desired controller for control system based on the strategy of finite abstraction.
Our future work will focus on this issue.

\begin{figure}[t]
\begin{center}
\centerline{\includegraphics[scale=0.9]{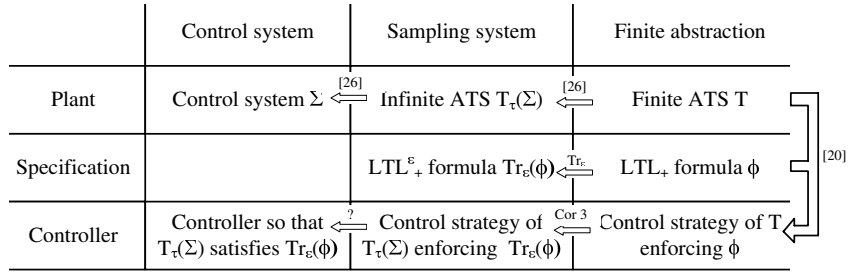}}\vspace{-0.6in}
\end{center}
\hfil\break
\caption{Controller design : disturbance case}\label{fig:graph2}
\end{figure}
\section{Conclusion}
\label{sec:conclusion}

This paper provides a modal characterization of alternating approximate bisimilarity.
Since alternating approximate bisimilarity is not always an equivalence relation, its modal characterization can not be provided in the usual style.
This paper introduces two relations over temporal logic ATL$_{\varepsilon}$ and adopt these relations to establish the desired modal characterization of alternating approximate bisimilarity in a new style (see Theorem~\ref{th:main}).
This result reveals a relationship between the approximate equivalence among alternating transition systems and the temporal logical properties satisfied by these systems.

Pola and Tabuada adopt alternating transition systems to model the samples of control systems with disturbances and their finite abstractions, and introduce the notion of alternating approximate bisimilarity to capture the equivalence between these systems~\cite{pola2008symbolic}\cite{pola2009symbolic}.
Based on the modal characterization of alternating approximate bisimilarity obtained in this paper, we provide the transformation function $Tr_\varepsilon$ from LTL$_+$-specifications to LTL$^\varepsilon_+$-specifications.
 Moreover, we show that, given a control system with disturbances, whose sampling system and finite abstraction are alternating approximate bisimilar, if the later realizes LTL$_+$-specification $\phi$ under control, then the former satisfies the corresponding LTL$^\varepsilon_+$-specification $Tr_\varepsilon(\phi)$ under control.
As illustrated in Fig~\ref{fig:graph2}, this result may be useful in designing the controller for  control systems with disturbances.
Future work will be devoted to perfecting the approach shown in Fig~\ref{fig:graph2}.

\bibliographystyle{spmpsci}
\bibliography{myref}

%
%

\end{document}